\newif\ifllncs\llncsfalse
\newif\ifanon\anonfalse

\ifllncs
\documentclass{llncs}
\else
\documentclass[11pt,letter]{article}
\fi 

\usepackage{breakcites}
\usepackage[utf8]{inputenc}
\usepackage[T1]{fontenc}
\ifllncs
\else
\usepackage{beton}
\usepackage{fullpage}
\fi

\usepackage{amsmath}
\usepackage{amsfonts}
\usepackage{amssymb}
\usepackage{mathtools, xparse}
\usepackage[dvipsnames]{xcolor}
\usepackage{graphicx}
\usepackage{braket}
\usepackage{url}
\usepackage{bm}
\usepackage{soul}
\usepackage{multirow}
\definecolor{DarkBlue}{RGB}{0,0,150}
\definecolor{NotSoDarkBlue}{RGB}{15,15,210}
\definecolor{DarkRed}{RGB}{150,0,0}
\definecolor{DarkGreen}{RGB}{0,100,0}
\usepackage[pdfstartview=FitH,pdfpagemode=UseNone,colorlinks,linkcolor=DarkRed,filecolor=blue,citecolor=DarkRed,urlcolor=DarkRed,pagebackref,breaklinks]{hyperref}
\usepackage[ruled, algosection]{algorithm2e}

\usepackage{orcidlink}

\def\dist{\mathsf{dist}}

\newcommand{\RR}{\mathbb{R}}

\newcommand{\ZZ}{\mathbb{Z}}

\DeclareMathOperator{\poly}{poly}

\allowdisplaybreaks

\DeclarePairedDelimiter{\norm}\lVert\rVert

\newcommand{\rvline}{\hspace*{-\arraycolsep}\vline\hspace*{-\arraycolsep}}

\newcommand{\condqed}{\ifllncs \qed \else \fi}

\newcommand{\condparagraph}[1]{\ifllncs \subsubsection{#1} \else \paragraph{#1} \fi}

\ifllncs
\else
\usepackage{amsthm}
\newtheorem{theorem}{Theorem}
\newtheorem{conjecture}{Conjecture}
\newtheorem{lemma}[theorem]{Lemma}
\newtheorem{corollary}[theorem]{Corollary}
\newtheorem{definition}[theorem]{Definition}

\numberwithin{theorem}{section}
\numberwithin{conjecture}{section}
\numberwithin{problem}{section}
\fi

\newtheorem{maintheorem}[theorem]{Main Theorem}

\newif\ifnotes
\notesfalse

\ifllncs
\author{Seyoon Ragavan\inst{1}\orcidlink{0009-0007-9628-2258}\and Vinod Vaikuntanathan\inst{1}\orcidlink{0000-0002-2666-0045}}
\institute{MIT CSAIL}
\ifanon\author{Anonymous Submission to CRYPTO 2024}\institute{}\fi 
\titlerunning{Layout Graphs, Random Walks and the $t$-wise Independence of SPN}
\authorrunning{Liu, Pelecanos, Tessaro and Vaikuntanathan}
\date{}
\else
\author{
Seyoon Ragavan
\\ MIT \\ \texttt{sragavan@mit.edu}
\and
Vinod Vaikuntanathan
\\ MIT \\ \texttt{vinodv@mit.edu}
}
\fi

\title{Space-Efficient and Noise-Robust\\ Quantum Factoring}

\ifllncs\else
\date{\today} \fi

\begin{document}
\maketitle

\begin{abstract}
   We provide two improvements to Regev's quantum factoring algorithm (Journal of the ACM 2025), addressing its space efficiency and its noise-tolerance. 
    
   Our first contribution is to improve the quantum space efficiency of Regev's algorithm while keeping the circuit size the same. Our main result  constructs a quantum factoring circuit using $O(n \log n)$ qubits and $O(n^{3/2} \log n)$ gates. We achieve the best of Shor and Regev (upto a logarithmic factor in the space complexity): on the one hand, Regev's circuit requires $O(n^{3/2})$ qubits and $O(n^{3/2} \log n)$ gates, while Shor's circuit requires $O(n^2 \log n)$ gates but only $O(n \log n)$ qubits. As with Regev, to factor an $n$-bit integer $N$, we run our circuit independently $O(\sqrt{n})$ times and apply Regev's classical postprocessing procedure. 

    Our optimization is achieved by implementing efficient and reversible exponentiation with Fibonacci numbers in the exponent, rather than the usual powers of 2, adapting work by Kaliski (arXiv:1711.02491) from the classical reversible setting to the quantum setting. 
    This technique also allows us to perform quantum modular exponentiation that is efficient in both space and size without requiring significant precomputation, a result that may be useful for other quantum algorithms. A key ingredient of our exponentiation implementation is an efficient circuit for a function resembling \emph{in-place} quantum-quantum modular multiplication. This implementation works with only black-box access to any quantum circuit for \emph{out-of-place} modular multiplication, which we believe is yet another result of potentially broader interest.

    Our second contribution is to show that Regev's classical postprocessing procedure can be modified to tolerate a constant fraction of the quantum circuit runs being corrupted by errors. In contrast, Regev's analysis of his classical postprocessing procedure requires all $O(\sqrt{n})$ runs to be successful. In a nutshell, we achieve this using lattice reduction techniques to detect and filter out corrupt samples.
\end{abstract}

\thispagestyle{empty}
\newpage
\thispagestyle{empty}
\tableofcontents
\newpage 
\pagenumbering{arabic}


\section{Introduction} 
\label{sec:intro}


Shor's landmark result from 1994~\cite{shor97} showed us how to factor numbers in quantum polynomial time. In particular, he constructed an $O(n^{2+\epsilon})$-size quantum circuit that uses $O(n)$ qubits (including ancilla qubits\footnote{For ease of exposition, we assume in the first part of this introduction that quantum factoring circuits are using the multiplication circuit by~\cite{kahanamokumeyer2024fast} which uses $O_\epsilon(n^{1+\epsilon})$ gates and $o(n)$ ancilla qubits.}), such that $O(1)$ runs of his circuit followed by polynomial-time classical postprocessing suffices to factor an $n$-bit integer. A number of followup works~\cite{beckman,vedral,seifert2001using,Cop02,CleveW00,beauregard,takahashi,zalka2006shors,DBLP:conf/pqcrypto/EkeraH17,gidney2017factoring,hrs17,gidney2019,DBLP:journals/quantum/GidneyE21,kahanamokumeyer2024fast} improved the Shor circuit in several ways, for example by constructing a log-depth implementation~\cite{CleveW00} or reducing the number of qubits to a mere $\approx 1.5n$~\cite{zalka2006shors}.

Polynomial improvements to the size of Shor's circuit remained elusive for nearly three decades, until Regev~\cite{Regev23} recently demonstrated an $O(n^{1.5+\epsilon})$-size quantum circuit for factoring. Regev shows how to factor $n$-bit integers using $\approx \sqrt{n}$ (parallel) runs of his factoring circuit. While the total number of operations remains the same, the hope is that the smaller Regev circuit (or a future modification of it) is easier to build and that it resists decoherence noise better. 

However, Regev's factoring algorithm (that is, his quantum circuit together with the classical postprocessing algorithm) has two key limitations, the first being its space (in)efficiency and the second its noise (in)tolerance.

In a nutshell, our work provides algorithms addressing both of these problems. Both are simple and independent plug-and-play modifications to Regev's algorithm~\cite{Regev23}. We reduce the space complexity to $O(n)$ by modifying the modular arithmetic part of Regev's quantum circuit, and we improve the tolerance to quantum errors by only modifying Regev's classical postprocessing procedure. We next discuss each of these problems in some more detail, and then outline our techniques and results.


\condparagraph{Space Complexity.}
A key efficiency consideration in the construction of quantum circuits is their space complexity, namely, the number of qubits including the ancilla qubits used by the circuit. (For more discussion on the importance of reducing space complexity in certain quantum computing architectures, e.g. superconducting architectures, see \cite{DBLP:journals/quantum/GidneyE21,gidneycomment}.) The number of qubits used in the original Shor circuit was $O(n)$. Several works~\cite{beckman, vedral, beauregard,takahashi,zalka2006shors,gidney2017factoring,hrs17,gidney2019, kahanamokumeyer2024fast} optimized the space complexity down to $\approx 1.5n$ qubits. However, none of these optimizations seems to apply to Regev's circuit whose space complexity stands at $O(n^{1.5})$.

To explain why, we start by reminding the reader that a key step in Shor's algorithm is the computation of the map $\ket{x} \mapsto \ket{a^x \bmod{N}}$ in superposition, where $a$ is a fixed base. The fact that $a$ is fixed allows one to precompute its powers $a, a^2,\ldots,a^{2^j},\ldots$ classically.
Once this is done, exponentiation in superposition can be done with space $O(n)$ using the multiplier by~\cite{kahanamokumeyer2024fast} and multiplying together the appropriate subset of the precomputed values. However, using this technique appears to prevent us from improving the size of the circuit beyond the $O(n^2)$ barrier essentially realized by Shor's algorithm; in the case of Regev's algorithm, we would have to precompute $a_i^{2^j} \bmod N$ for $\Theta(\sqrt{n})$ values of $i$ and $\Theta(\sqrt{n})$ values of $j$. This is hence $\Theta(\sqrt{n} \cdot \sqrt{n} \cdot n) = \Theta(n^2)$ bits of precomputed information; a quantum circuit using all of these values would hence be expected to require $\Omega(n^2)$ gates. 

Instead, if one tries to implement the map $\ket{x}\ket{a} \mapsto \ket{x}\ket{a^x \bmod{N}}$ via the classic square-and-multiply algorithm, one runs into the conundrum of implementing the squaring circuit modulo $N$, as observed by Eker{\aa} and Gidney~\cite{gidneycomment}. On the one hand, doing this reversibly and in-place seems as hard as factoring $N$. On the other hand, writing the squared result to a new register consumes extra space; indeed, this is the source of the $O(n^{3/2})$ space requirement for Regev's circuit \cite{Regev23}.
This state of affairs raises a natural question:
    \begin{quote}
    \begin{center}
      {\em Can we achieve the best of Shor and Regev?}
      \end{center}
    \end{quote} 
\noindent
Concretely, can we construct a quantum factoring circuit with $O(n^{1.5+\epsilon})$ gates and $O(n)$ ancilla qubits?

\condparagraph{Error Tolerance.} Quantum decoherence noise is a central and ubiquitous obstacle to realizing quantum circuits in practice that achieve functionality that we cannot already perform classically~\cite{martinis2019, DBLP:journals/quantum/GidneyE21, sitannisq, cai23}. For the specific case of factoring, it appears that the associated error correction procedures are both necessary and costly. Cai~\cite{cai23} shows that Shor's algorithm breaks down in the presence of uncorrected noise, even if one restricts to very small and structured noise in only the final QFT part of the circuit. Additionally, Gidney and Eker{\aa}~\cite{DBLP:journals/quantum/GidneyE21} analyze the concrete cost of factoring 2048-bit integers using Shor's algorithm while correcting for quantum errors throughout, and show that even for carefully chosen parameters, this would take 8 hours with 20 million physical qubits, a significant overhead relative to the costs of running Shor's circuit in the absence of errors.

How much error correction does one need to factor? In the case of Shor's algorithm~\cite{Shor94}, a circuit with $O(n^{2+\epsilon})$ gates is run $O(1)$ times. For all of these runs to work, the error probability per logical gate would have to be brought down to $O(1/n^{2+\epsilon})$. While it initially appears that the situation with Regev~\cite{Regev23} would be better because of the smaller circuit, this is actually not the case upon a closer look. Indeed, Regev's circuit comprises of $O(n^{1.5+\epsilon})$ gates and is run $O(n^{0.5})$ times. Regev's analysis of his classical postprocessing only applies in the case where \emph{every} run of the circuit is successful. Hence the logical error probability per logical gate would still have to be the same $\tilde{O}(1/n^{2+\epsilon})$.
The natural question that arises is:
    \begin{quote}
    \begin{center}
      {\em Can we achieve better noise-tolerance than} both {\em Shor} and {\em Regev?}
      \end{center}
    \end{quote} 
\noindent
Indeed, there is a natural opportunity to improve the noise-tolerance in the case of Regev's algorithm~\cite{Regev23}: if only the classical postprocessing could be modified to tolerate a constant fraction of unsuccessful runs, then the per-gate error threshold would only have to be $O(1/n^{1.5+\epsilon})$, improving on both Shor's and Regev's algorithms.

We now proceed to describe our results for quantum factoring in more detail.


\subsection{Our Results for Optimizing Space}\label{sec:introspaceresults}

\begin{table}
\begin{center}
    \begin{tabular}{|c|c|c|c|}
    \hline
    Algorithm & Mult. algorithm & Qubit count & Gate count  \\
    \hline\hline
    Shor \cite{shor97} & Near-linear~\cite{Harvey21} & $O(n \log n)$ & $O(n^2\log n)$ \\
    Optimized Shor & Quantum Toom-Cook & \multirow{2}{*}{$\bm{(2+o(1))n}$} & \multirow{2}{*}{$O_\epsilon(n^{2+\epsilon})$} \\
    \cite{beauregard, takahashi, gidney2017factoring, hrs17}  & \cite{gidney2019, kahanamokumeyer2024fast} & & \\
    Optimized Shor~\cite{zalka2006shors} & Schoolbook & $\bm{(1.5 + o(1))n}$ & $O(n^3)$ \\
    \hline
    Regev's algorithm \cite{Regev23} & Near-linear~\cite{Harvey21} & $O(n^{1.5})$ & $\bm{O(n^{1.5}\log n)}$ \\
    \hline
    Our optimization of Regev & Near-linear~\cite{Harvey21} & $O(n \log n)$ & $\bm{O(n^{1.5}\log n)}$ \\
    \multirow{2}{*}{Our optimization of Regev} & Quantum Toom-Cook & \multirow{2}{*}{$\bm{(10.32+o(1))n}$} & \multirow{2}{*}{$O_\epsilon(n^{1.5 + \epsilon})$} \\
    & \cite{gidney2019, kahanamokumeyer2024fast} & & \\
    \hline
    \end{tabular}
    \caption[Comparison of our asymptotic space optimizations with previous work]{Comparison of our results with previous work. Asymptotically best-known results for either space or size are highlighted in bold. Note, importantly, that all values here are just for \textit{one run of the circuit}. They do not account for the fact that a circuit for Shor's algorithm only requires $O(1)$ independent runs, while a circuit for Regev's algorithm as well as  ours requires $\sqrt{n}+4$ independent runs. The constant $\epsilon$ can be made arbitrarily close to 0 and arises from the multiplication circuits by~\cite{kahanamokumeyer2024fast}. The number of qubits in the last line is derived from Corollary~\ref{cor:resultforgidneymult} assuming that the constant $C$ in Regev's number-theoretic assumption is $1+o(1)$.}
    \label{summarytable}
\end{center}
\end{table}

Regev's algorithm~\cite{Regev23} achieves a smaller circuit of $O(n^{1.5+\epsilon})$ gates in comparison to Shor's $O(n^{2+\epsilon})$~\cite{Shor94}; however, for reasons discussed earlier it requires at least $\approx 3n^{1.5}$ qubits\footnote{In more detail, using the notation in Section \ref{sec:setup}, Regev's circuit requires at least $\log_2 D \cdot n \approx An^{1.5}$ qubits. The calculation in Appendix \ref{sec:paramcheck} together with observations in Section \ref{sec:regevconjecture} implies that we should take $A = 3 + o(1)$.} as opposed to the $O(n)$ qubits in Shor and optimized implementations of it. In this work, we construct a quantum circuit that asymptotically achieves the best of both worlds: it has a size of $O(n^{1.5+\epsilon})$ matching Regev's circuit; and it requires $O(n)$ qubits, nearly matching what is known for optimized versions of Shor's circuit. We compute the concrete number of qubits required (see our Theorem~\ref{thm:mainresult} below) which, using schoolbook multiplication, leads us to $(10.32+o(1))n$ qubits and $O(n^{1.5+\epsilon})$ gates, vis-a-vis the results of \cite{beckman, vedral, beauregard,takahashi,zalka2006shors,gidney2017factoring,hrs17, kahanamokumeyer2024fast} which use $\approx 2n$ qubits and $O(n^{2+\epsilon})$ gates, and the result of \cite{zalka2006shors} which uses $\approx 1.5n$ qubits and $O(n^3)$ gates. Our space improvement on Regev's circuit~\cite{Regev23} is substantial not only asymptotically but also for cryptographically relevant problem sizes; when $n = 2048$, our circuit uses at least $\approx 13 \times$ fewer qubits than Regev's circuit. (See Table~\ref{summarytable} for a full list of our results and comparison to prior work.)

%

As discussed earlier, approaches to modular exponentiation via precomputation or repeated squaring do not appear to yield the best-of-both-worlds results we seek. We avoid these issues with the key technique of Fibonacci exponentiation, a method of exponentiation that avoids modular squaring and instead relies solely on modular multiplication. Previous works have considered using Fibonacci numbers rather than powers of two in the exponent for particular applications~\cite{Byrne2007, Klein2008, Meloni}. Additionally, Kaliski~\cite{kal17} explicitly shows how to use this Fibonacci technique to achieve efficient modular exponentiation in the setting of \emph{classical} reversible computation.

We show that these ideas can be adapted to the quantum setting (see Section \ref{sec:kaliski} for a more detailed comparison of our work with~\cite{kal17}). While Fibonacci exponentiation seems to not be of much use in Shor's circuit, it turns out to be quite powerful in optimizing the space usage of Regev's circuit. As noted by Kaliski~\cite{kal17}, it also offers a general way to efficiently compute modular exponents on quantum computers in cases where precomputation is impossible or inefficient; we discuss this in Appendix \ref{sec:modexp}.

We begin by stating our main theorem for qubit complexity, and compare the result to Regev~\cite{Regev23} and optimizations of Shor's algorithm~\cite{shor97,beckman, vedral, beauregard, hrs17, takahashi,gidney2017factoring,gidney2019,kahanamokumeyer2024fast}. In a nutshell, the algorithm by \cite{Regev23} requires $O(n^{1.5})$ qubits, which we bring down to $O(n)$ while retaining the $O(n^{1.5+\epsilon})$ quantum circuit size. Optimizations of Shor's algorithm that retain the $O(n^{2+\epsilon})$ circuit size require $\approx 2n$ qubits~\cite{kahanamokumeyer2024fast}; thus, compared to these results, our algorithm works with the same number of qubits but achieves a $O(\sqrt{n})$ factor saving in the circuit size (just as \cite{Regev23} does). We state our result in terms of an arbitrary multiplication circuit; previously in this introduction, we restricted attention to the multiplication circuit by~\cite{kahanamokumeyer2024fast} for convenience.
\begin{maintheorem}\label{thm:mainresult}
    Assume there is a quantum circuit that implements the operation $$\ket{a} \ket{b} \ket{t} \ket{0^S} \mapsto \ket{a} \ket{b} \ket{(t+ab) \bmod N} \ket{0^S}$$ with $G$ gates where $N, a, b, t$ are all $n$-bit integers with $0 \leq a, b, t < N$ and $2^{n-1} \leq N < 2^n$, and $S$ here is the number of ancilla qubits.\footnote{We do not make any assumptions about the circuit's behavior on inputs where any of $a, b, t$ are  $\geq N$.}
    
    Under conjecture~\ref{conjecture} (a number-theoretic assumption proposed by~\cite{Regev23}), there exists a classical polynomial-time algorithm that outputs a non-trivial factor of $N$ using $\sqrt{n} + 4$ calls to a quantum circuit on $$S + \bigg(\frac{C+2}{\log \phi} + 6 + o(1)\bigg)n \approx S + (1.44C + 8.88)\cdot n$$ qubits using $O(n^{1/2} \cdot G + n^{3/2})$ gates, where $\phi = (1+\sqrt{5})/2$ is the golden ratio. Here, $C$ is the absolute constant from conjecture~\ref{conjecture}.
\end{maintheorem}

We now state three corollaries of this general theorem by plugging in different known results for integer multiplication. Some of these multiplication results do not have the specific structure Theorem~\ref{thm:mainresult} requires e.g. they may be classical rather than quantum, or multiply integers over $\ZZ$ rather than modulo $N$. However, it turns out that all these can be readily adapted to our setting; we discuss this in detail in Appendix \ref{sec:multimplementation}.

Using the classical $O(n \log n)$ multiplication algorithm due to \cite{Harvey21}  allows us to set $G = S = O(n \log n)$ and obtain the following corollary:
\begin{corollary}\label{cor:resultforfastmult}
    Under the same assumption as in Theorem \ref{thm:mainresult}, there is a quantum circuit for factoring that uses $O(n \log n)$ qubits and $O(n^{3/2} \log n)$ gates.
\end{corollary}

If we want to push the number of qubits needed all the way down to $O(n)$, we can employ the space-efficient quantum implementation of Karatsuba's algorithm \cite{karatsuba} due to Gidney~\cite{gidney2019} which uses $S = O(n)$ ancilla qubits and $G = O(n^{\log_2 3})$ gates. Alternatively, we can use the recent result by Kahanamoku-Meyer and Yao~\cite{kahanamokumeyer2024fast} which combines ideas by Draper~\cite{draper} with Toom-Cook multiplication~\cite{cookmult, toom1963complexity} to obtain a multiplication circuit using $S = O_\epsilon(\log n)$ ancilla qubits and $G = O_\epsilon(n^{1+\epsilon})$ gates\footnote{We note that the multiplication circuit by~\cite{kahanamokumeyer2024fast} already has exactly the form we need; we do not need to incur any overheads from compiling it into the form we need using our results in Appendix \ref{sec:multimplementation}.} for any $\epsilon > 0$, yielding the following corollary:
\begin{corollary}\label{cor:resultforgidneymult}
    Under the same assumption as in Theorem \ref{thm:mainresult}, for any $\epsilon > 0$, there is a quantum circuit for factoring that uses $$\bigg(\frac{C+2}{\log \phi} + 6 + o_\epsilon(1)\bigg)n \approx (1.44C + 8.88)\cdot n$$ qubits and $O_\epsilon(n^{3/2 + \epsilon})$ gates, where $C$ is the constant from Conjecture \ref{conjecture}.
\end{corollary}
This latter result achieves the same asymptotic space complexity as Shor's algorithm using the same multiplication circuit, except with $O(\sqrt{n})$ factor smaller number of gates. We compare these results with previous work on quantum circuits for factoring in Table \ref{summarytable}. We also detail the various sources of space cost in our algorithm and potential areas for further optimizations in Section \ref{sec:furtheropt}.

As we discuss in Section \ref{sec:regevconjecture}, following a heuristic argument by Regev~\cite{Regev23} suggests that taking $C = 1+\epsilon$ should be sufficient for conjecture~\ref{conjecture} to hold. Plugging this into the statement of Theorem \ref{thm:mainresult} tells us that the space complexity of our circuit should be $\approx S + (10.32+o(1))n$.

\subsection{Our Results on Error-Correction}\label{sec:introecresults} 

%

We show that Regev's classical postprocessing~\cite{Regev23} can be modified to tolerate a constant fraction of unsuccessful runs of the quantum circuit, thus only requiring a $O(1/n^{1.5+\epsilon})$ bound on the probability of logical error per logical gate.

\begin{maintheorem}\label{thm:eccmainthm}
    (Informal, see Section \ref{sec:formalecc} and Theorem \ref{thm:mainthm} for formal statement) Assume the probability of error in one run of Regev's circuit is a sufficiently small constant $p > 0$. Then there exists a classical polynomial-time algorithm that, given $m = \Omega(\sqrt{n})$ potentially corrupt samples from Regev's quantum circuit, outputs a non-trivial factor of $N$ with probability $\Omega(1)$. 
\end{maintheorem}

Our algorithm is conceptually simple. As we will explain in Section \ref{sec:regevreview}, uncorrupted samples from Regev's circuit will essentially be random samples from $\mathcal{L}^\ast/\ZZ^d$, where $\mathcal{L}$ is a lattice depending on the integer $N$ that needs to be factored. For simplicity, assume that corrupted samples are uniform from the torus $\RR^d/\ZZ^d$. Then the main idea is that a small linear combination of samples that includes even one corrupt sample will also be uniform and hence should not be close to the dual lattice. Based on this, we devise a filtering procedure using lattice reduction techniques to detect and filter out corrupt samples. This gives us a collection of uncorrupted samples, which can now be plugged into Regev's classical postprocessing procedure~\cite{Regev23} as-is.

However, we also extend our results beyond the case where the corrupt samples are uniform over $\RR^d/\ZZ^d$. We define a formal model in terms of a general error distribution $\mathcal{D}$ and clearly state a general condition on $\mathcal{D}$ that suffices for our filtering procedure to go through. We defer the formal models and statements to Section \ref{sec:formalecc}.

\subsection{Related Work}
\label{sec:kaliski}

\condparagraph{On Kaliski's Work and Quantum-Quantum Multiplication.} We compare our techniques to achieve our space-complexity result with those of Kaliski~\cite{kal17}.

At its core, the beautiful work of Kaliski~\cite{kal17} constructs classical reversible and space-efficient algorithms for modular exponentiation i.e. the map $(1, a, z) \mapsto (1, a, z, a^z\bmod{N})$. 
(Kaliski also then shows how to ``forget'' $a$, thereby obtaining $(1, 1, z, a^z\bmod{N})$, but this is neither useful nor applicable to our situation because these additional results assume the algorithm knows the order $\varphi(N)$ of the group, which we do not.)

While Kaliski's algorithm is reversible in the classical world, it is not directly applicable to quantum computation. The key arithmetic operation used by Kaliski is the mapping $(a, b) \mapsto (a, ab \bmod{N})$. This is classically reversible assuming $\gcd(a, N) = 1$, but it is not trivial to implement this on a quantum computer.\footnote{We note that Kaliski does consider applications of his Fibonacci exponentiation idea to quantum computation~\cite{kaliski2017quantum, kal17}, but does not appear to identify or address the problem of implementing $(a, b) \mapsto (a, ab\bmod{N})$ on a quantum computer.} Indeed, it is at least as difficult as inversion modulo $N$: the inverse of a circuit computing this mapping would map $(a, b) \mapsto (a, a^{-1}b \bmod{N})$. Taking $b = 1$ would immediately yield a circuit for modular inversion. The well-known circuits for modular inversion based on the extended Euclidean~\cite{ProosZalka03} or binary GCD~\cite{roetteler17} algorithms use $O(n^2)$ gates. There are lesser-known circuits~\cite{schonhageeuc, thull1990uni, Moller2008OnSA} for modular inversion that only use $\widetilde{O}(n)$ gates, but are still significantly more complicated and (concretely) expensive than carrying out an $n$-bit multiplication. Either way, the cost of modular inversion poses an undesirable bottleneck to quantumly implementing the map $(a, b) \mapsto (a, ab\bmod{N})$.

This problem is often referred to as \emph{in-place quantum-quantum modular multiplication}~\cite{rines2018high}. ``In-place'' specifies that the output $ab\bmod{N}$ should overwrite the input register $b$, and ``quantum-quantum'' specifies that $a$ and $b$ are both in superposition. In contrast, ``quantum-classical'' multiplication would be when $a$ is a classical constant. This is a more straightforward task, and was addressed in Shor's original paper~\cite{shor97}.

Previous work by Rines and Chuang~\cite{rines2018high} on quantum modular multiplication constructs a very careful white-box implementation of schoolbook multiplication using Montgomery multipliers with $O(n^2)$ gates and $O(n)$ qubits. Their circuits can be used for out-of-place quantum-quantum multiplication or in-place quantum-classical multiplication, but do not appear adaptable to in-place quantum-quantum multiplication without using the costly extended Euclidean algorithm to coherently compute inverses~\cite{ProosZalka03}.\footnote{We thank Martin Eker{\aa} and Joel G{\"a}rtner for pointing out to us that the results by~\cite{rines2018high} do not easily handle in-place quantum-quantum multiplication.} Moreover, their method does not appear to be adaptable to arbitrary multiplication algorithms, and in particular, does not match state-of-the-art classical algorithms using $\widetilde{O}(n)$ gates~\cite{Harvey21, schonhage1971fast}.

Our algorithm approaches the problem of in-place quantum-quantum multiplication by \emph{sidestepping} the cost of inversion modulo $N$. We start by replacing each register $a$ in Kaliski's procedure with the tuple $(a, a^{-1}\bmod{N})$, so that our goal is now to implement the map $(a, a^{-1}\bmod{N}, b, b^{-1}\bmod{N}) \mapsto (a, a^{-1}\bmod{N}, ab\bmod{N}, (ab)^{-1}\bmod{N})$. In Lemma \ref{lemma:beauregardmult}, we show that this map can be implemented using $O(1)$ black-box calls to \emph{any} quantum circuit for \emph{out-of-place} quantum-quantum modular multiplication. This allows us to utilize any classical multiplication algorithm in our quantum circuit, including those by~\cite{Harvey21} or~\cite{schonhage1971fast} (after applying the necessary procedures in Appendix \ref{sec:multimplementation}).

\condparagraph{On the Work of \cite{ekeragartner}.} 
Subsequent to a public posting of our work containing only our space complexity result, 
Eker{\aa} and G{\"a}rtner~\cite{ekeragartner} showed how to adapt both Regev's algorithm~\cite{Regev23} and our space optimization of it to obtain quantum circuits for the discrete logarithm problem over $\ZZ_p$ that achieve an analogous improvement over Shor~\cite{shor97}. Second, concurrent to our results on error-tolerance, they also address the question of improving the error tolerance of Regev, but in a different way; instead of first detecting and filtering out corrupt samples like we do, they show under a different assumption on the distribution of corrupt samples that Regev's classical postprocessing~\cite{Regev23} as-is can withstand a constant fraction of samples being corrupted.

While the algorithm arising from their result is hence conceptually simpler than ours, it relies on a stronger assumption. For example, in the case that the corrupted samples are uniform from the torus $\RR^d/\ZZ^d$, we are able to prove correctness of our algorithm assuming nothing more than the same Conjecture \ref{conjecture} needed by~\cite{Regev23}, namely that there exists one short vector in a lattice $\mathcal{L}$ that gives rise to a nontrivial square root mod $N$. On the other hand, the assumption by~\cite{ekeragartner} which depends on both $\mathcal{L}$ and the distribution of errors requires at the very least that $\mathcal{L}$ has a short basis, which is a stronger assumption. For this reason, we view both the analysis by~\cite{ekeragartner} and our result as distinct and important contributions towards achieving robust and efficient quantum factoring.

We note that our result described in Section \ref{sec:formalecc} on detecting errors also readily adapts to the discrete logarithm algorithm by~\cite{ekeragartner}. We provide an outline in Appendix \ref{sec:eccdl}.

\condparagraph{Other Follow-Up Work.} Subsequent to postings of both our work and that by~\cite{ekeragartner}, it was shown by Pilatte~\cite{pilatte2024unconditional} that a variant of our space-efficient implementation of Regev's factoring algorithm~\cite{Regev23} and its adaptation to discrete logarithms by~\cite{ekeragartner} can be proven to be unconditionally correct, at the expense of polylogarithmic factors in the qubit and gate complexity. Specifically, it is shown by~\cite{pilatte2024unconditional} that if we allow the bases in Regev's circuit~\cite{Regev23} to be as large as $\exp(\widetilde{O}(\sqrt{n}))$ rather than just $\poly(n)$, we can prove the analogue of Conjecture \ref{conjecture} and thus obtain an \emph{unconditionally correct} quantum factoring algorithm that makes $O(\sqrt{n})$ calls to a quantum circuit that uses $O(n^{3/2}\log^3n)$ gates and $O(n\log^2n)$ qubits. While this algorithm is likely undesirable in practice for cryptographically relevant problem sizes because of its use of larger bases, the result by Pilatte~\cite{pilatte2024unconditional} still serves to increase confidence in conjecture \ref{conjecture}, and hence the correctness of Regev's algorithms~\cite{Regev23} and its variants in~\cite{ekeragartner, cryptoeprint:2024/636} and this work.

In another follow-up work~\cite{cryptoeprint:2024/636}, we carefully analyze the prefactor in the gate count of our space-efficient circuit. We provide methods to substantially improve the prefactor in the number of gates, as well as to lower the prefactor in the space complexity from $\approx 10.32$ to $9+\epsilon$ for any $\epsilon > 0$, by working with more general sequences of integers than the Fibonacci numbers.

\condparagraph{Organization of the Paper.} We start with an overview of Regev's factoring algorithm in Section~\ref{sec:regevreview}. We follow this up with a formal statement of our results in Section~\ref{sec:mainresults}, the space-efficiency result in Section~\ref{sec:oracle}, and the error-tolerance result in Section~\ref{sec:error-resilience}.

\section{Setup and Notation}\label{sec:setup}

We retain all notation from \cite{Regev23} and restate it here for convenience. Let $N < 2^n$ be an $n$-bit number. Let $d = \lfloor\sqrt{n}\rfloor$ and $b_1, \ldots, b_d$ be some small $O(\log n)$-bit integers (e.g. $b_i$ is the $i$th prime number) and let $a_i = b_i^2 \bmod{N}$. For any integer $t$, let $[t]$ denote the set $\left\{1, 2, \ldots, t\right\}$. We use $\log$ to denote the base-2 logarithm throughout this paper. Also, let $\phi$ denote the golden ratio.
Regev's algorithm uses a number of parameters:
\begin{itemize}
\item Let $C > 0$ be an absolute constant given by conjecture \ref{conjecture}; 
\item Let $A > C$ be another constant we specify later. For Regev's factoring algorithm and our space optimization, we will take $A = C+2+o(1)$, but a larger constant $A$ may be needed in order to be compatible with our error detection results (see Section \ref{sec:formalecc} for details);
\item Let $R = 2^{(A+o(1))\sqrt{n}}$;   and
\item Let $D$ be a power of 2 in $[2\sqrt{d} \cdot R, 4\sqrt{d} \cdot R]$. Note that $D$ is also $2^{(A+o(1))\sqrt{n}}$.  These are the same parameters $R$ and $D$ defined by \cite{Regev23}. 
\end{itemize}
Regev defines the following lattices in $d$ dimensions:
\begin{align*}
    \mathcal{L} &= \left\{(z_1, \ldots, z_d) \in \ZZ^d: \prod_{i = 1}^d a_i^{z_i} \equiv 1 \bmod{N}\right\}\text{, and} \\
    \mathcal{L}_0 &= \left\{(z_1, \ldots, z_d) \in \ZZ^d: \prod_{i = 1}^d b_i^{z_i} \equiv \pm 1 \bmod{N}\right\} \subseteq \mathcal{L}.
\end{align*}
We will also work closely with the dual lattice $$\mathcal{L}^\ast = \left\{y \in \RR^d: \langle x, y \rangle \in \ZZ\text{ }\forall x \in \mathcal{L}\right\} \supseteq \ZZ^d.$$

\section{Overview of Regev's Factoring Algorithm}\label{sec:regevreview}

Regev's algorithm~\cite{Regev23} starts by ruling out simple possibilities where factoring is easy i.e. if $N$ is even, a prime power, or shares a common factor with any of the $b_i$. For the remainder of this paper, we assume none of these is the case.

The goal of Regev's algorithm~\cite{Regev23} is to find one vector $z = (z_1, \ldots, z_d)$ in $\mathcal{L}\backslash \mathcal{L}_0$. Given such a vector, one can construct $b = \prod_{i = 1}^d b_i^{z_i}\bmod{N}$. Since $z \in \mathcal{L}$, $b$ must be a square root of 1 modulo $N$, moreover it must be a nontrivial square root of 1 since $z \notin \mathcal{L}_0$. Hence $N$ divides $(b-1)(b+1)$ but not either term individually, so $\gcd(b-1, N)$ will be a nontrivial divisor of $N$.

Regev's algorithm can be broken down into two distinct pieces. First, there is a quantum circuit that essentially produces one noisy sample from $\mathcal{L}^\ast/\ZZ^d$. This circuit is run $d+4$ times. Secondly, there is a classical postprocessing procedure that takes these samples and essentially recovers a basis for all vectors in $\mathcal{L}$ of norm at most $T = 2^{C\sqrt{n}}$. Under a heuristic assumption made by~\cite{Regev23}, at least one of these vectors will belong to $\mathcal{L}\backslash\mathcal{L}_0$.

\subsection{Overview of Regev's Quantum Circuit}

We give an overview of Regev's quantum factoring circuit, just enough to understand our results; for more details, we refer the reader to the original work~\cite{Regev23}.
To help keep track of space usage, all lemmas will clearly specify the number of ancilla qubits needed (if any).

\subsubsection{Constructing a Superposition over Vectors $z$}\label{sec:discretegaussian}

For $s > 0$, define the Gaussian function $\rho_s: \RR^d \rightarrow \RR$ as
$$\rho_s(z) = \exp(-\pi||z||^2/s^2).$$
Then in this step, the algorithm constructs a discrete Gaussian state $\ket{\psi}$ within $1/\text{poly}(d)$ trace distance of the state proportional to:
$$\sum_{z \in \left\{-D/2, \ldots, D/2-1\right\}^d} \rho_R(z) \ket{z}.$$
The complexity of this step can be summarized in the following lemma from \cite{Regev23}.
\begin{lemma}
    $\ket{\psi}$ can be constructed in-place (i.e., without any ancilla qubits) with $$O(d (\log D + \poly (\log d))) = O(n)$$ gates. The number of qubits needed to store the vector $\ket{z}$ is $d \log D = (A+o(1))n$.
\end{lemma}
\begin{proof}
    We restate the outline by~\cite{Regev23} here and then explain why this can be done without ancilla qubits. We refer the reader to~\cite{grover2002creating, regev09} for details. Firstly, it suffices to compute a one-dimensional discrete Gaussian state over scalars $z \in \left\{-D/2, \ldots, D/2-1\right\}$ in-place using $O(\log D + \poly(\log d))$ gates (since our $d$-dimensional discrete Gaussian is a tensor product of $d$ one-dimensional discrete Gaussian states).

    As explained by~\cite{Regev23}, this can be achieved (up to an error of $1/\poly(d)$ trace distance) by adapting a standard procedure as in~\cite{grover2002creating, regev09}. For the $O(\log d)$ most significant qubits, one can apply the appropriate rotation to that qubit conditioned on the values of the previous qubits, using $\poly(\log d)$ gates. For the remaining qubits, the correct rotation will be very close to the $\ket{+}$ state, so we can simply apply Hadamard gates in-place to all of them. (This is where the $1/\poly(d)$ trace distance error is incurred.) This takes at most $\log D$ gates.

    To address the space consideration, observe that computing the $O(\log d)$ most significant qubits requires $\poly(\log d)$ ancilla qubits. These ancilla qubits simply store the angles for each rotation, which is a classical function of the bits of $z$ and can hence be uncomputed. After uncomputation, these ancilla qubits will be back in the $\ket{0}$ state, and can be \emph{reused} as low-order qubits for the discrete Gaussian state (since we need $\approx \log D \gg \poly(\log d)$ lower-order qubits).
\condqed \end{proof}

\subsubsection{A Quantum Oracle to Compute $\prod_{i = 1}^d a_i^{z_i+D/2}\bmod{N}$}

Next, the algorithm computes (in superposition) \begin{equation}\label{eqn:exp}\prod_{i = 1}^d a_i^{z_i+D/2}\pmod{N}
\end{equation}
(the offset by $D/2$ in the exponent is to make all exponents non-negative).

Recall that we use $G$ to denote the number of gates used by our $n$-bit multiplication circuit; see the statement of Theorem \ref{thm:mainresult} for a precise description. Then, Regev achieves this step by first constructing a classical circuit with $O(\log D \cdot G) = O(n^{1/2} \cdot G)$ gates for this computation, then implementing this quantumly. Classical circuits can be ``compiled'' into quantum circuits via standard techniques~\cite{bennett_logical_1973, bennett_timespace_1989, levine_note_1990}, which we detail for completeness in Appendix \ref{sec:classicaltoquantum}.

Regev's classical circuit uses a repeated squaring procedure, while also exploiting the fact that the $a_i$'s are small to reduce the number of large-integer multiplications required. After compiling this into a quantum oracle and making some minor optimizations,\footnote{In particular, Regev does not quite just take his classical circuit (for computing equation~\ref{eqn:exp}) and black-box convert it into quantum, which would result in $O(n^{3/2}\log n)$ space complexity (as explained in Appendix \ref{sec:classicaltoquantum}). However, he is able to reuse the ancilla qubits used for multiplications, resulting in just  $O(n^{3/2})$ qubits.} Regev obtains the following lemma: 

\begin{lemma}\label{lemma:regevoracle}
    As in Theorem \ref{thm:mainresult}, let $G$ and $S$ be the number of gates and ancilla qubits respectively for our multiplication circuit on $n$-bit integers. Then there exists a quantum circuit mapping $$\ket{z}\ket{0^M} \mapsto \ket{z} \Ket{\prod_{i = 1}^d a_i^{z_i+D/2} \bmod{N}} \ket{\psi}.$$ Here, $$M = S + O(n^{3/2})$$ is the initial number of ancilla qubits, and $\ket{\psi}$ is some possibly nonzero state on $M-n$ qubits. Moreover, this circuit uses $O(n^{1/2} \cdot G)$ gates.
\end{lemma}

We note that this step is the performance bottleneck in Regev's algorithm both in terms of gates and qubits. One of our main contributions is finding a different way to implement this oracle that only requires $S + O(n)$ qubits while retaining the same asymptotic gate complexity. We discuss our improvement in more detail in Section \ref{sec:spaceoptoverview}.

\subsubsection{Measurement and QFT}

The final stage of the quantum part of Regev's algorithm~\cite{Regev23} first measures the register storing $\prod_{i = 1}^d a_i^{z_i+D/2}\pmod{N}$ to collapse the $\ket{z}$ register to a superposition over some coset from $\ZZ^d/\mathcal{L}$. Then the algorithm applies the inverse of the circuit in Lemma \ref{lemma:regevoracle} to uncompute the nonzero state $\ket{\psi}$.

Now the algorithm applies an approximate QFT \cite{Cop02} modulo $D$ i.e. over $\ZZ_D^d$ to the $\ket{z}$ register, measures, and divides by $D$ to obtain a vector close to a uniform sample from the dual lattice $\mathcal{L^\ast}/\ZZ^d$. More formally, with probability $1-1/\poly(d)$, we obtain a sample of the form $w_i = v_i + \delta_i$. Here, $v_i$ is a uniform sample from $\mathcal{L^\ast}/\ZZ^d$ and $\delta_i$ is some error of magnitude at most $2^{(-A+o(1))n/d}$. The sample $w_i$ is the final output of each run of the quantum circuit.

To understand the complexity of this part of the quantum circuit, what we need is the following lemma from \cite{Cop02,Regev23}.

\begin{lemma}
    The approximate QFT can be computed in-place with circuit size $$O(d \log D(\log \log D + \log d)) = O(n \log n)$$ gates. The space usage here is the number of qubits needed to store $\ket{z}$, which is $d \log D = (A+o(1))n$ qubits.
\end{lemma}

\subsection{Overview of Regev's Classical Postprocessing}

Regev's classical postprocessing procedure~\cite{Regev23} works with the lattice $\mathcal{L}$, but it can really be viewed as an algorithm for any lattice $\Lambda \subseteq \ZZ^d$. We state Regev's result in these terms below:

\begin{lemma}\label{lemma:regevpostprocessing}
    Let $\Lambda \subseteq \ZZ^d$ be a lattice, and let $m = d+4$. Additionally, let $T > 0$ be some norm bound. Assume we are given as input $m$ independent samples of the form $$w_i = v_i + \delta_i,$$ where each $v_i$ is a uniform sample from $\Lambda^\ast/\ZZ^d$ and $\delta_i$ is some additive error of magnitude at most $\delta$. Additionally, assume the following inequality: $$(m+d)^{1/2} \cdot 2^{(m+d)/2} \cdot (m+1)^{1/2} \cdot T < \delta^{-1} \cdot (4 \det \Lambda)^{-1/m}/6.$$Then there exists a classical polynomial-time algorithm that, with probability at least $1/4$, outputs a finite sequence of vectors $z_1, z_2, \ldots, z_l \in \Lambda$ such that any $u \in \Lambda$ with $||u||_2 \leq T$ can be written as an integer linear combination of the $z_i$'s.

    We note that the algorithm is deterministic; the success probability is taken over the randomness of the $w_i$'s.
\end{lemma}

For the factoring algorithm,~\cite{Regev23} takes $\mathcal{L} = \Lambda$ and $T = 2^{C\sqrt{n}}$ and uses the bound that $\det \mathcal{L} \leq N < 2^n$. Plugging in and solving implies that we require $A \geq C+2+o(1)$ (see Appendix \ref{sec:paramcheck} for this simple calculation).

We remark here that our other main contribution in this paper is extending the algorithm in Lemma \ref{lemma:regevpostprocessing} to handle cases where a small constant fraction of the $m$ samples may be completely corrupted. Our algorithm uses Regev's classical postprocessing procedure~\cite{Regev23} as its final step, but we need to modify the analysis by~\cite{Regev23} to obtain a result slightly different from Lemma \ref{lemma:regevpostprocessing}. We defer a discussion of this algorithm and its analysis to Appendix \ref{sec:regevlatticeanalysis}. 

\subsection{Regev's Number-Theoretic Assumption}\label{sec:regevconjecture}

Regev's choice of the parameter $T$ arises from the following heuristic number-theoretic assumption~\cite{Regev23}:

\begin{conjecture}\label{conjecture} 
There exists a vector in $\mathcal{L} \backslash \mathcal{L}_0$ of $\ell_2$ norm at most $T = 2^{C\sqrt{n}}$, for some given constant $C > 0$.
\end{conjecture}

As observed by \cite{Regev23}, a simple pigeonhole principle argument shows that there exists a nonzero vector in $\mathcal{L}$ with norm at most $2^{(1+o(1))\sqrt{n}}$. Following Regev's heuristic argument suggests that taking $C = 1+\epsilon$ should be sufficient.

Under this assumption, there exists some nonzero vector $u \in \mathcal{L}\backslash \mathcal{L}_0$ of $\ell_2$ norm at most $T$. By Lemma \ref{lemma:regevpostprocessing}, $u$ is expressible as an integer linear combination of $z_1, z_2, \ldots, z_l$. Since $u \notin \mathcal{L}_0$, there exists some $i \in [l]$ such that $z_i \notin \mathcal{L}_0 \Rightarrow z_i \in \mathcal{L} \backslash \mathcal{L}_0$. Hence at least one of $z_1, \ldots, z_l$ is an element of $\mathcal{L} \backslash \mathcal{L}_0$, so we can simply try the aforementioned gcd calculation for each of them one at a time.

Finally, we note that follow-up work by Pilatte~\cite{pilatte2024unconditional} shows that if we allow the bases $b_1, \ldots, b_d$ to be as large as $\exp(O(d \log d))$, then the analogue of Conjecture \ref{conjecture} can be proven unconditionally, thus providing an unconditionally correct variant of Regev's algorithm that improves asymptotically on Shor's algorithm. This completes our overview of Regev's factoring algorithm.

\section{Our Improvements to Regev's Algorithm}
\label{sec:mainresults} 

\subsection{Reducing the Number of Qubits}\label{sec:spaceoptoverview}

Here, we formally state our space improvement to Lemma \ref{lemma:regevoracle}, thus yielding Theorem \ref{thm:mainresult}.

\begin{lemma}\label{lemma:oracle}
    (Compare with Lemma \ref{lemma:regevoracle}) As in Theorem \ref{thm:mainresult}, let $G$ and $S$ be the number of gates and ancilla qubits respectively for our multiplication circuit on $n$-bit integers. Then there exists a quantum circuit mapping $$\ket{z}\ket{0^M} \mapsto \Ket{\prod_{i = 1}^d a_i^{z_i+D/2} \bmod{N}} \ket{\psi_z}.$$ Here, $$M = S+\bigg(\frac{A}{\log \phi}-A+6+o(1)\bigg)n$$ is the initial number of ancilla qubits, and $\ket{\psi_z}$ is some possibly nonzero state on $M+An-n$ qubits that could depend on $z$. Moreover, this circuit uses $O(n^{1/2} \cdot G + n^{3/2})$ gates. Therefore, the total number of qubits used (i.e. the space usage) is $$S + \bigg(\frac{A}{\log \phi} + 6 + o(1)\bigg)n.$$
\end{lemma}

\subsection{Tolerating Quantum Errors}\label{sec:formalecc}

Regev's algorithm assumes that all quantum gates and qubits exist without error in every run of the circuit. This may not be the case in practice; indeed, the difficulty of detecting and correcting quantum errors is a significant barrier to constructing quantum computers at scale~\cite{fowler2012surface,campbell2019applying,DBLP:journals/quantum/GidneyE21}. We introduce a model for such errors and state our results formally.

\subsubsection{The Error Model}\label{sec:errormodel}
We model errors by assuming that an additional Hamming error is applied to each sample from the quantum circuit with constant probability. Concretely, let $\mathcal{D}$ be a noise distribution over $\RR^d/\ZZ^d$, $\eta \in \left\{0, 1\right\}$ an integer, and $p = \Theta(1)$ the probability of a quantum error in one run of the circuit. Instead of sampling $w_i = v_i + \delta_i$ directly, we will assume that the following corruption procedure is applied to it:
\begin{enumerate}
    \item With probability $1-p$, we observe the sample $w_i = v_i + \delta_i$. In this case, we say the sample is \emph{uncorrupted}.
    \item With probability $p$, we obtain the following \emph{corrupted} sample:
    \begin{enumerate}
        \item Sample $\epsilon_i \leftarrow \mathcal{D}$. This sampling is independent for each $i$ (this is a reasonable assumption since each run of the circuit is independent).
        \item The sample we observe is now $w_i = \eta(v_i + \delta_i) + \epsilon_i$.
    \end{enumerate}
\end{enumerate}
We comment on the role of the parameter $\eta$ in this model:
\begin{itemize}
    \item If $\eta = 1$, then we are dealing with additive errors: the sample we expect to see is perturbed by a sample from $\mathcal{D}$.
    \item If $\eta = 0$, then we are dealing with overwrite errors: we directly see a sample from $\mathcal{D}$.
\end{itemize}
Our results will hold for either choice of $\eta$, and hence either error model. We note that, in the important special case where $\mathcal{D}$ is uniform over $\RR^d/\ZZ^d$, the two possibilities for $\eta$ yield the same distribution. Moreover, we remark that the overwrite error model is extremely general, and subsumes the additive error model. For example, suppose we let $\mathcal{D}_1$ be the distribution of $v_i + \delta_i$ (i.e. an output from Regev's circuit~\cite{Regev23} without any quantum errors) and $\mathcal{D}_2$ some independent noise distribution over $\RR^d/\ZZ^d$, so that the output of a corrupt run of Regev's circuit would be the sum of a sample $v_i + \delta_i$ from $\mathcal{D}_1$ and a sample $\epsilon_i$ from $\mathcal{D}_2$. The natural way to view this distribution is using our additive error model, with $\mathcal{D} = \mathcal{D}_2$. However, we could also view this using the overwrite error model, by simply taking $\mathcal{D}$ to sample independently from $\mathcal{D}_1$ and $\mathcal{D}_2$ and sum the results (equivalently, we are taking the convolution of $\mathcal{D}_1$ and $\mathcal{D}_2$). (We remark that we do not require the existence of efficient or even inefficient algorithms, classical or quantum, for sampling from $\mathcal{D}$; we just need the distribution $\mathcal{D}$ to exist for this model to make sense.) 

This then begs the question of why we are considering the additive error model at all. The reason for this is that our results will require a special property of the distribution $\mathcal{D}$, that we refer to as being ``well-spread''. We define this formally in Definition \ref{defn:distassumption}. It may in some cases be easier to verify this ``well-spread'' condition when working with the additive error model than with the overwrite error model. For instance, in the example provided in the above paragraph, it will typically be easier to prove that $\mathcal{D}_2$ is well-spread than proving that the convolution of $\mathcal{D}_1$ and $\mathcal{D}_2$ is well-spread (due to the additional noise that comes from the $\delta_i$'s in $\mathcal{D}_1$).

\subsubsection{Our Results}\label{sec:errformalresults}

Let $\alpha > 0$ and $\gamma \in (0, \alpha-1)$ be constants that we will specify later. We will run Regev's quantum circuit $m = \alpha d$ times. 

Our results will need to depend in some way on the structure of the noise distribution $\mathcal{D}$, which we capture in the following definition which, in words, says that a distribution $\mathcal{D}$ is well-spread with respect to a lattice if there is no way to take small (non-zero) integer linear combinations (i.e. linear combinations with small integer coefficients) of samples from $\mathcal{D}$ that results in a point close to the lattice.

\begin{definition}\label{defn:distassumption}
    Let $\dist(\cdot, \cdot)$ denote the distance on the torus $\RR^d/\ZZ^d$ i.e. mod 1, and let $\Lambda \subseteq \ZZ^d$ be a lattice. Then we say $\mathcal{D}$ is $(\alpha,\gamma,A)$-\emph{well-spread with respect to} a lattice $\Lambda^\ast$ if the following holds for any positive integer $k \leq (\alpha-\gamma-1) d$: 
    For i.i.d. $\epsilon_1, \ldots, \epsilon_k \leftarrow \mathcal{D}$, with probability $1-o(1)$, there do not exist integers $a_1, a_2, \ldots, a_k$ such that:
    \begin{itemize}
        \item $|a_i| \leq 2^{\frac{(\alpha-\gamma+3+o(1))n}{2d}}$ for all $i$.
        \item $a_i \neq 0$ for at least one index $i$.
        \item $\dist\big(\sum_{i = 1}^k a_i\epsilon_i, \Lambda^\ast/\ZZ^d \big) \leq 2^{\frac{(-2A+\alpha-\gamma+3+o(1))n}{2d}}.$
    \end{itemize}

    
    \noindent
    Informally, there is a reasonable chance (over the choice of $k$ samples from $\mathcal{D}$) that no non-zero linear combination of the samples is very close to $\Lambda^\ast/\ZZ^d$.
\end{definition}

\noindent
To justify this definition, we observe that this holds for the uniform distribution on the torus $\RR^d/\ZZ^d$:
\begin{lemma}\label{lemma:uniformiswellspread}
    If $A > \frac{(\alpha-\gamma)(\alpha-\gamma+3) + 2}{2}$, the uniform distribution on $\RR^d/\ZZ^d$ is $(\alpha,\gamma,A)$-well-spread with respect to $\Lambda^\ast$.
\end{lemma}
\begin{proof}
    In a nutshell, this follows from a union bound and volume argument. We defer the details of the calculation to Appendix \ref{sec:uniformiswellspread}.
\condqed \end{proof}

\noindent
We can now formally state our result for detecting and handling errors in a constant fraction of runs in Regev's circuit. We first formulate this in terms of an arbitrary lattice $\Lambda \subseteq \ZZ^d$:

\begin{lemma}\label{lemma:ourpostprocessing}
    (Compare with Lemma \ref{lemma:regevpostprocessing}) Let $\Lambda \subseteq \ZZ^d$ and $n$ be a positive integer such that $\det \Lambda < 2^n$. Additionally, let $m = \alpha d$ and let $T > 0$ be some norm bound. Assume we are given potentially corrupt samples $w_1, \ldots, w_m$ constructed according to the error model in Section \ref{sec:errormodel} with the following specifications:
    \begin{itemize}
        \item The initial vectors $v_1, \ldots, v_m$ are independent uniform samples from $\Lambda^\ast/\ZZ^d$.
        \item The small perturbations $\delta_i$ have magnitude at most $\delta$.
        \item The errors $\epsilon_i$ are sampled from a noise distribution $\mathcal{D}$ on $\RR^d/\ZZ^d$. (Recall that we see the uncorrupted sample $w_i = v_i + \delta_i$ with probability $1-p$ and a corrupted sample $w_i = \eta(v_i + \delta_i) + \epsilon_i$ with probability $p$.)
    \end{itemize}
    Let $\gamma > 0$ be a constant such that all of the following inequalities hold:

    \begin{itemize}
        \item $p < 1 - (\gamma+1)/\alpha$. Note that this implicitly requires $\gamma < \alpha-1$.
        \item $A \geq \frac{\alpha-\gamma+3}{2}+o(1).$
        \item $\big(\frac{e\alpha}{\alpha-\gamma}\big)^{\alpha-\gamma} \cdot 2^{-\gamma+1} < 1.$ 
    \end{itemize}
    Additionally, assume the following ``special inequality'': 
    $$(\gamma d + d)^{1/2} \cdot 2^{(\gamma+1)d/2} \cdot (\gamma d + 1)^{1/2} \cdot T < \delta^{-1} \cdot (4 \det \Lambda)^{-1/(\gamma d)} \cdot 2^{-\alpha/\gamma}/6.$$

    \noindent
    Also, assume $\mathcal{D}$ is $(\alpha, \gamma,A)$-well-spread with respect to $\Lambda^\ast$. Then there exists a classical polynomial-time algorithm that, with probability at least $1/4-o(1)$, outputs a finite sequence of vectors $z_1, z_2, \ldots, z_l \in \Lambda$ such that any $u \in \Lambda$ with $||u||_2 \leq T$ can be written as an integer linear combination of the $z_i$'s.

    We note that the algorithm is deterministic; the success probability is taken over the randomness of the $w_i$. 
\end{lemma}

\noindent
For example, we can select the parameters $\alpha = 7$ and $\gamma = 5.9$, which then require $p < 1/70$. We briefly comment on the role of each of the stated constraints on the parameters in Lemma \ref{lemma:ourpostprocessing}:
\begin{itemize}
    \item The first inequality is to limit the number of corrupted samples.
    \item The second inequality is needed for Lemma \ref{lemma:gainonesample} and ultimately places a lower bound on the radius of the initial discrete Gaussian superposition one needs to construct (as in Section \ref{sec:discretegaussian}).
    \item The third inequality is needed in Appendix \ref{sec:acciaroproof}; our algorithm will ultimately select some subset of $\gamma d$ samples and we need to take a union bound over all possible choices of this subset, necessitating an upper bound on $\binom{\alpha d}{\gamma d}$.
    \item Finally, the special inequality arises from analysis in Appendix \ref{sec:regevlatticeanalysis} (based on that by~\cite{Regev23}) and places an additional lower bound on the Gaussian radius; intuitively, it says that the small additive errors $\delta_i$ that appear in both uncorrupted and corrupted samples cannot be too large.
\end{itemize}

To apply Lemma \ref{lemma:ourpostprocessing} in the specific case of Regev's factoring circuit~\cite{Regev23}, we may once again take $\Lambda = \mathcal{L}$ and note that $\det \mathcal{L} < 2^n$, then take $T = 2^{C\sqrt{n}}$ and $\delta = 2^{(-A+o(1))n/d}$ and solve the special inequality. This is a straightforward calculation, which we defer to Appendix \ref{sec:ourparamcheck}. Combined with the observations in Section \ref{sec:regevconjecture}, this yields the following theorem.

\begin{theorem}\label{thm:mainthm}
    Assume Conjecture \ref{conjecture}, and let $\gamma > 0$ be a constant such that the hypotheses of Lemma \ref{lemma:ourpostprocessing} (except for the special inequality) hold for the lattice $\mathcal{L}$. Additionally, assume $A \geq \max(C + \frac{\gamma^2 + \gamma + 2}{2 \gamma}, \frac{\alpha-\gamma+3}{2}) + o(1)$. Then there exists a polynomial-time algorithm that, given $m = \alpha d$ samples from Regev's quantum circuit corrupted according to the error model in Section \ref{sec:errormodel}, outputs a non-trivial factor of $N$ with probability $1/4-o(1)$.
\end{theorem}

\noindent
In the special case where $\mathcal{D}$ is the uniform distribution over $\RR^d/\ZZ^d$, combining Theorem \ref{thm:mainthm} with Lemma \ref{lemma:uniformiswellspread} implies that one should take $$A \geq \max\left(C + \frac{\gamma^2 + \gamma + 2}{2 \gamma}, \frac{(\alpha-\gamma)(\alpha-\gamma+3) + 2}{2}\right) + o(1)~.$$

As mentioned earlier, our algorithm and analysis in Lemma \ref{lemma:ourpostprocessing} are also directly applicable to the recent adaptation by Eker{\aa} and G{\"a}rtner~\cite{ekeragartner} of Regev's factoring algorithm~\cite{Regev23} to the discrete logarithm problem modulo prime $p$. We provide an outline in Appendix \ref{sec:eccdl}.

\section{Reducing Qubits via Space-Efficient Exponentiation}\label{sec:oracle}

In this section, we construct a quantum oracle that computes the desired mapping in Lemma~\ref{lemma:oracle}.

\subsection{Reversible Fibonacci Exponentiation}\label{sec:fibexp}

This is the key idea used by Kaliski~\cite{kal17} for space-efficient and reversible classical exponentiation; we restate it here and explain its applicability to our optimization. The bottleneck in \cite{Regev23}'s construction of this oracle is in the use of repeated squaring. As a simple example, suppose we have some $n$-bit integer $\ket{a}$ and we wish to compute $\ket{a^{2^k} \bmod N}$ for some $k$. The natural classical-inspired approach would be to use repeated squaring as follows:
$$\ket{a} \mapsto \Ket{a^2 \bmod N} \mapsto \Ket{a^4 \bmod N} \mapsto \ldots \mapsto \Ket{a^{2^k} \bmod N}.$$
The problem is that none of these operations are easily reversible; computing square roots mod $N$ is hard, and there is also the fact that squaring mod $N$ is not one-to-one. Thus carrying out this computation requires storing all of these quantities in separate registers. This leads to an $O(k)$ blow-up in the space complexity, which appears wasteful.

To get around this, we use the fact that it is much easier to implement reversible multiplication than squaring. Concretely, if we could implement an operation resembling $\ket{a}\ket{b} \mapsto \ket{a} \ket{ab\bmod N}$, we could proceed as follows:
\begin{align*}\ket{a}\ket{a} \mapsto \ket{a}\ket{a^2\bmod{N}} & \mapsto \ket{a^3\bmod{N}}\ket{a^2\bmod{N}} \\
& \mapsto \ket{a^3\bmod{N}}\ket{a^5\bmod{N}} \mapsto \ldots
\end{align*}
We now find Fibonacci numbers rather than powers of 2 in the exponent. For this reason, previous works have also considered using Fibonacci numbers rather than powers of 2 in the exponent for fast and reversible exponentiation \cite{Byrne2007, kal17, Klein2008, Meloni}.

Notice that when using this idea, we will have an extra register lingering around. For example, if our goal is to compute $\ket{a^5 \bmod N}$ in one register, this would give us $\ket{a^3 \bmod N}$ in the other register. We would like to be able to clean this up if needed, but fortunately this is straightforward; we can copy the final output to a new register, and then uncompute our circuit to clean up all intermediate qubits. This will just leave us with the inputs and outputs of our computation. We will use this idea repeatedly in our algorithm.

If we could actually implement $\ket{a}\ket{b} \mapsto \ket{a}\ket{ab\bmod{N}}$, we could directly apply Algorithm \textsc{FibExp} from~\cite{kal17} (with a straightforward adaptation to computing a product of multiple exponents rather than one exponent). This task is known as ``in-place quantum-quantum modular multiplication''. (See Section \ref{sec:kaliski} for definitions of these terms.) The issue is that this appears difficult to do on a quantum computer for arbitrary multiplication algorithms, to the best of our knowledge. Rines and Chuang~\cite{rines2018high} show that this operation can be done on a quantum computer in $O(n^2)$ gates with a careful implementation of schoolbook multiplication using Montgomery multipliers. However, we would like to be able to cleanly adopt and use any classical multiplication algorithm in our circuit, for example to improve asymptotic gate complexity.

To achieve this, we show how to implement an ``abstract form'' of in-place quantum-quantum modular multiplication in the next two lemmas, using only black-box access to a multiplication circuit. As in Theorem \ref{thm:mainresult}, we assume throughout that this multiplication circuit maps $$\ket{a}\ket{b}\ket{t}\ket{0^S} \mapsto \ket{a}\ket{b}\ket{(t + ab) \bmod N} \ket{0^S}$$ where $a$, $b$, $t$ are all $n$-bit integers and $0 \leq a, b, t < N$. $S$ hence denotes the number of ancilla qubits used by the circuit, and we also use $G$ to denote the number of gates in this circuit. In other words, this circuit implements \emph{out-of-place} quantum-quantum modular multiplication.

Our lemmas will also make use of some ``dirty ancilla qubits'' that may not necessarily be in the $\ket{0}$ state. This will allow us to reuse qubits from other parts of our algorithm and cut down on the number of qubits needed in our circuit. The idea of using dirty ancilla qubits to achieve space savings was introduced by \cite{hrs17} in relation to implementing Shor's algorithm with fewer qubits.

The first lemma we need is essentially due to Shor \cite{shor97}. Using the terminology introduced in Section \ref{sec:kaliski}, this addresses the task of in-place quantum-classical modular multiplication.

\begin{lemma}\label{lemma:beauregardeasy}
    Let $a \in [0, N-1]$ be an integer coprime to $N$. Then there exists a circuit using $O(G+n)$ gates mapping $$\ket{x} \ket{0^{S+n}}\ket{g} \mapsto \ket{ax \bmod N} \ket{0^{S+n}} \ket{(-a^{-1}g)\bmod N}$$ for any integers $x, g$ that are reduced mod $N$.
    
    Note that this computation uses and restores $S+n$ clean ancilla qubits, while applying some reversible transformation to $n$ dirty ancilla qubits that initially store the state $\ket{g}$.
\end{lemma}
\begin{proof}
    We can classically precompute $a^{-1} \bmod N$ efficiently using the extended Euclidean algorithm. Now proceed as follows using our quantum multiplication circuit given in Theorem \ref{thm:mainresult}:
    \begin{align*}
        \ket{x} & \ket{0^n} \ket{g} \ket{0^S} \\ 
        &\rightarrow \ket{x} \ket{a} \ket{g} \ket{0^S}\text{ (writing in a classical constant using bit-flips)} \\
        &\rightarrow \ket{x} \ket{a} \ket{(g+ax) \bmod N} \ket{0^S} \\
        &\rightarrow \ket{x} \ket{-a^{-1} \bmod N} \ket{(g+ax) \bmod N} \ket{0^S} \\
        & \hspace{1.5in} \text{ (writing in a classical constant again)} \\
        &\rightarrow \ket{(x + (-a^{-1} \cdot (g+ax))) \bmod N} \ket{-a^{-1} \bmod N} \ket{(g+ax) \bmod N} \ket{0^S} \\
        & \hspace{1in} = \ket{(-a^{-1}g) \bmod N} \ket{-a^{-1} \bmod N} \ket{(g+ax) \bmod N} \ket{0^S} \\
        &\rightarrow \ket{(-a^{-1}g) \bmod N} \ket{a} \ket{(g+ax) \bmod N} \ket{0^S} \\ 
        & \hspace{1.5in} \text{ (writing in a classical constant again)} \\
        &\rightarrow \ket{(-a^{-1}g) \bmod N} \ket{a} \ket{(g+ax + (-a^{-1}g \cdot a)) \bmod N} \ket{0^S} \\
        &\hspace{1in} =  \ket{(-a^{-1}g) \bmod N} \ket{a} \ket{ax \bmod N} \ket{0^S} \\
        &\rightarrow \ket{(-a^{-1}g) \bmod N} \ket{0^n} \ket{ax \bmod N} \ket{0^S}\\ 
        & \hspace{1.5in} \text{ (writing in a classical constant again)} \\
        &\rightarrow \ket{ax \bmod N}\ket{0^n} \ket{(-a^{-1}g) \bmod N} \ket{0^S}
    \end{align*}
    This runs our multiplication circuit three times and does $O(n)$ bit flips and bit swaps, for a total of $O(G+n)$ gates. This completes our proof.
\condqed \end{proof}
In the next lemma, we adapt this idea to the quantum-quantum case i.e. where $a$ may be a superposition of integers rather than a classical constant. Note that this lemma does not exactly implement quantum-quantum modular multiplication, but rather an abstracted form of it where we view $\ket{a}\ket{a^{-1}\bmod{N}}$ as an abstract representation of $a$.
\begin{lemma}\label{lemma:beauregardmult}
    There exists a quantum circuit using $O(G+n)$ gates such that, for all $n$-bit integers $a, b, g \in [0, N-1]$ such that $a$ and $b$ are coprime to $N$, it will map\footnote{We are not concerned about the circuit's behavior on other basis states.} 
    \begin{align*} 
    \ket{a}\ket{a^{-1}\bmod N} & \ket{b}\ket{b^{-1}\bmod N}\ket{g}\ket{0^{S}} \\ 
    & \mapsto \ket{a}\ket{a^{-1}\bmod N}\ket{ab\bmod N}\ket{(ab)^{-1}\bmod N}\ket{g}\ket{0^{S}}.
    \end{align*}
    Note that this computation uses and restores $S$ clean ancilla qubits and $n$ dirty ancilla qubits.
\end{lemma}

\begin{proof}
    The quantum multiplication circuit assumed in Theorem \ref{thm:mainresult} allows us the following computation when $0 \leq t < N$:
    $$\ket{a}\ket{b}\ket{t}\ket{0^S} \mapsto \ket{a}\ket{b}\ket{(t+ab)\bmod N}\ket{0^S}.$$
    We can also hence use the inverse of this circuit, which will behave as follows when $0 \leq t < N$:
    $$\ket{a}\ket{b}\ket{t}\ket{0^S} \mapsto \ket{a}\ket{b}\ket{(t-ab)\bmod N}\ket{0^S}.$$
    We can use these two circuits to proceed as follows. (All the computations below are $\bmod{N}$, which we omit for compactness. For example, when we say $a^{-1}$, we mean $a^{-1}\bmod N$.)
    \begin{align*}
        \ket{a} & \ket{a^{-1}} \ket{b}\ket{b^{-1}} \ket{g}\ket{0^S} \\
        & \rightarrow \ket{a}\ket{a^{-1}}\ket{b}\ket{b^{-1}} \ket{g+ab}\ket{0^S} \\
        & \rightarrow \ket{a}\ket{a^{-1}}\ket{b - (a^{-1} \cdot (g+ab))}\ket{b^{-1}} \ket{g+ab}\ket{0^S} \\
        & \hspace{.5in} = \ket{a}\ket{a^{-1}}\ket{-a^{-1}g}\ket{b^{-1}} \ket{g+ab}\ket{0^S} \\
        & \rightarrow \ket{a}\ket{a^{-1}}\ket{-a^{-1}g }\ket{b^{-1}} \ket{(g+ab)+a \cdot (-a^{-1}g)} \ket{0^S} \\
        & \hspace{.5in} =\ket{a}\ket{a^{-1}}\ket{-a^{-1}g }\ket{b^{-1}} \ket{ab}\ket{0^S} \\
        & \rightarrow \ket{a}\ket{a^{-1}}\ket{-a^{-1}g + a^{-1}b^{-1}}\ket{b^{-1}} \ket{ab }\ket{0^S} \\
        & \rightarrow \ket{a}\ket{a^{-1}}\ket{-a^{-1}g + a^{-1}b^{-1}}\ket{b^{-1} - (a \cdot (-a^{-1}g + a^{-1}b^{-1}))} \ket{ab}\ket{0^S} \\
        &\hspace{.5in} = \ket{a}\ket{a^{-1}}\ket{-a^{-1}g + a^{-1}b^{-1}}\ket{g} \ket{ab}\ket{0^S} \\
        & \rightarrow \ket{a}\ket{a^{-1}}\ket{(-a^{-1}g + a^{-1}b^{-1}) + a^{-1} \cdot g}\ket{g} \ket{ab}\ket{0^S} \\
        &\hspace{.5in} = \ket{a}\ket{a^{-1}}\ket{(ab)^{-1} }\ket{g} \ket{ab}\ket{0^S} \\
        & \rightarrow \ket{a}\ket{a^{-1}}\ket{ab}\ket{(ab)^{-1}}\ket{g}\ket{0^S}.
    \end{align*}
    This runs our multiplication circuit four times and its inverse twice and then does a constant number of swaps of $n$-bit registers at the end, for a total of $O(G+n)$ gates. This completes our proof.
\condqed \end{proof}

The next lemma essentially plugs Lemma \ref{lemma:beauregardmult} into Kaliski's algorithm~\cite{kal17}, while taking care to optimize space usage. We quickly set up some notation. Define the Fibonacci numbers by $F_0 = 0, F_1 = 1$, and then $F_k = F_{k-1} + F_{k-2}$ for $k \geq 2$. Let $K$ be maximal such that $F_K \leq D$. We have:
\begin{align*}
    K &= (1+o(1)) \cdot \frac{\log D}{\log \phi} 
    = (1+o(1)) \cdot \frac{(A+o(1))\sqrt{n}}{\log \phi} 
    = (\alpha + o(1)) \sqrt{n},
\end{align*}
where we let $\alpha = \frac{A}{\log \phi}$.

Also, let $\ket{\psi(a)}$ denote the state $\ket{a}\ket{a^{-1}\bmod N}$ for an $n$-bit integer $a$ coprime to and reduced modulo $N$. Each $\ket{\psi(a)}$ hence requires $2n$ qubits to store.

\begin{lemma}\label{lemma:fibonacci}
    \cite{kal17} There exists a quantum circuit using $O(K(G+n)) = O(n^{1/2}(G+n))$ gates such that, for all $n$-bit integers $c_1, \ldots, c_K$ coprime to and reduced modulo $N$, it will map \begin{align*} \ket{\psi(c_1)}\ldots \ket{\psi(c_K)} & \ket{0^{S+5n}} \\ 
    & \mapsto \ket{\psi(c_1)}\ldots \ket{\psi(c_K)} \Ket{\psi\left(\prod_{j = 2}^K c_j^{F_{j-1}}\right)} \Ket{\psi\left(\prod_{j = 1}^K c_j^{F_j}\right)} \ket{0^{S+n}}.
    \end{align*}
    \noindent
    (All registers throughout this lemma and its proof work with integers mod $N$, so we drop the $\bmod\ N$ everywhere for convenience.)
\end{lemma}

\begin{proof}
    We will use $4n$ of our ancilla qubits to store two states $\ket{\psi(x_1)}$ and $\ket{\psi(x_2)}$. Thus the state at any given point in our algorithm will be $$\ket{\psi(c_1)}\ket{\psi(c_2)}\ldots \ket{\psi(c_K)} \ket{\psi(x_1)} \ket{\psi(x_2)} \ket{0^{S+n}}.$$ We will update the values of $x_1$ and $x_2$ throughout the algorithm.

    Lemma \ref{lemma:beauregardmult} tells us that we can perform the operation $\ket{\psi(x)} \ket{\psi(y)} \ket{0^{S+n}} \mapsto \ket{\psi(x)} \ket{\psi(xy)} \ket{0^{S+n}}$ using $O(G+n)$ gates. We denote this as updating $y \leftarrow xy$ and will iteratively use this, as described in Algorithm \ref{algo:basicfibonacci}.
    \begin{algorithm}
        \KwData{Initial state $\ket{\psi(c_1)}\ket{\psi(c_2)}\ldots \ket{\psi(c_K)} \ket{0^{S+5n}}$}
        \KwResult{Final state $\ket{\psi(c_1)}\ket{\psi(c_2)}\ldots \ket{\psi(c_K)} \ket{\psi(\prod_{j = 2}^K c_j^{F_{j-1}})} \ket{\psi(\prod_{j = 1}^K c_j^{F_j})} \ket{0^{S+n}}$}
        \begin{enumerate}
            \item Initialize $x_1 \leftarrow 1$ and $x_2 \leftarrow 1$. This is just a matter of copying some qubits into ancilla qubits, and the state is now $\ket{\psi(c_1)}\ket{\psi(c_2)}\ldots \ket{\psi(c_K)} \ket{\psi(x_1)} \ket{\psi(x_2)} \ket{0^{S+n}} = \ket{\psi(c_1)}\ket{\psi(c_2)}\ldots \ket{\psi(c_K)} \ket{\psi(1)} \ket{\psi(1)} \ket{0^{S+n}}$.
            \item Repeat the following for $j = K, K-1, \ldots, 1$ in that order:
            \begin{itemize}
                \item Update $x_1 \leftarrow x_1x_2$ using Lemma \ref{lemma:beauregardmult}.
                \item Update $x_1 \leftarrow x_1c_j$ using Lemma \ref{lemma:beauregardmult}.
                \item Swap $x_1$ and $x_2$ (i.e. swap $\ket{\psi(x_1)}$ and $\ket{\psi(x_2)}$).
            \end{itemize}
        \end{enumerate}
        \caption{Fibonacci multi-exponentiation}\label{algo:basicfibonacci}
    \end{algorithm}

    Before we show the correctness of our algorithm, let us quickly analyze its complexity. The initialization takes $O(n)$ gates. Within the loop, there are $O(G+n)$ gates from calls to Lemma \ref{lemma:beauregardmult} and then $O(n)$ gates to swap $x_1$ and $x_2$. The loop has $K$ iterations, so the total number of gates is $O(K(G+n)) = O(n^{1/2}(G+n))$.

    To show that the algorithm runs correctly, we claim that for all $j \leq K-1$, at the end of round $j$, we will have $x_1 = \prod_{i = j+1}^K c_i^{F_{i-j}}$ and $x_2 = \prod_{i = j}^K c_i^{F_{i+1-j}}$. This will follow by a straightforward induction, which we present in Appendix \ref{sec:fibonacciproof}.
\condqed \end{proof}

Note, importantly, that we would not want to use Lemma \ref{lemma:fibonacci} directly in our algorithm; storing $c_1, \ldots, c_K$ at the same time is an undesirable space overhead. Instead, we will compute $c_1, \ldots, c_K$ on an as-needed basis and apply the same idea; this way, we only need to store one of them at any given time. We explain how to do this and how we define the $c_j$'s in the next section. 

\subsection{Combining Fibonacci Exponentiation with Regev's Optimization}

To use the results from Section \ref{sec:fibexp}, we want to decompose each of our exponents as a sum of distinct Fibonacci numbers. Concretely, recall that we want to compute the expression
$\prod_{i = 1}^d a_i^{z_i + D/2}.$
So for each $i \in [d]$, we would like to write:
$$z_i + D/2 = \sum_{j = 1}^K z_{i, j} F_j,$$
for $z_{i, j} \in \left\{0, 1\right\}$. It is well-known that this can be achieved with a simple greedy algorithm; see Appendix \ref{sec:zeckendorf} for details. Our first task is to compute these coefficients:

\begin{lemma}\label{lemma:zeckendorf}
    There exists a quantum circuit using $O(n^{3/2})$ gates mapping the state $$\ket{z} \ket{0^{dK-d\log D}} \ket{0^{O(\sqrt{n})}} \mapsto \ket{z_{i, j}: i \in [d], j \in [K]} \ket{0^{O(\sqrt{n})}}.$$ 
\end{lemma}
\begin{proof}
    We repeat the following greedy procedure for each $i \in [d]$. Note that integers here are computed in absolute terms, rather than modulo $N$. We need the ability to compute in-place additions and subtractions on $O(\sqrt{n})$-bit integers with $O(\sqrt{n})$ ancilla qubits; it was shown by~\cite{draper} that this is possible. We also need to be able to compare integers of length $O(\sqrt{n})$, but this need not be in-place so can also be done with $O(\sqrt{n})$ ancilla qubits.
    \begin{enumerate}
        \item Let $t$ denote the number in the register currently holding $z_i$. First update $t \leftarrow t + D/2$ (so that this register now holds $z_i + D/2$).
        \item Set aside $K$ ancilla qubits to hold $z_{i, j}$ for $j \in [K]$, so that $z_{i, j} = 0$ for all $j$ initially.
        \item Now for each $j = K, K-1, \ldots, 1$, check whether $t \geq F_j$ and write the output of this comparison to the qubit $z_{i, j}$. Then use $z_{i, j}$ as a control qubit to conditionally update $t \leftarrow t - F_j$.
        \item By Lemma \ref{lemma:zeckendorfgreedy}, this greedy algorithm will correctly decompose $z_i+D/2$ as a sum $\sum_{j = 1}^K z_{i, j} F_j$ of distinct Fibonacci numbers, and we will have $t = 0$ at the end. Hence we have freed up those $\log D$ bits as ancilla qubits to use in later steps.
    \end{enumerate}
    
    We have already observed that correctness follows from Lemma \ref{lemma:zeckendorfgreedy}. For the runtime, each step of the innermost loop over $j$ uses $O(\log D) = O(\sqrt{n})$ gates. So, each step of the outer loop over $i$ uses $O(K \log D) = O(n)$ gates. Finally, multiplying by $d = \sqrt{n}$ yields a gate complexity of $O(n^{3/2})$.

    Finally, we address space. All individual steps in the loop can clearly be done using $O(\sqrt{n})$ ancilla qubits (which we can then reuse). Other than that, each step of the loop consumes $K$ ancilla qubits but then frees up $\log D$ ancilla qubits. The total initial ancilla requirement is hence $d(K - \log D) + O(\sqrt{n})$ as desired.
\condqed \end{proof}
Now that we can calculate the $z_{i, j}$'s, let us write our desired expression in those terms:
\begin{align*}
    \prod_{i = 1}^d a_i^{z_i + D/2} &= \prod_{i = 1}^d \prod_{j = 1}^K a_i^{z_{i, j} F_j} 
    = \prod_{j = 1}^K \left(\prod_{i = 1}^d a_i^{z_{i, j}}\right)^{F_j}.
\end{align*}
We can calculate each $\prod_{i = 1}^d a_i^{z_{i, j}}$ efficiently following the idea by \cite{Regev23} to exploit the fact that the $a_i$'s are very small integers:
\begin{lemma}\label{lemma:regevmultiplication}
    \cite{Regev23} There exists a quantum circuit using $O(\sqrt{n} \log^3n)$ gates mapping $$\ket{t_1}\ldots\ket{t_d}\ket{0^{\widetilde{O}(\sqrt{n})}} \mapsto \ket{t_1}\ldots\ket{t_d}\Ket{\prod_{i = 1}^d a_i^{t_i}}\ket{0^{\widetilde{O}(\sqrt{n})}}.$$ 
    Here, the $t_i \in \left\{0, 1\right\}$ for all $i$.
\end{lemma}
\begin{proof}
    \cite{Regev23} shows that such a computation can be done classically with $O(d\log^3d) = O(\sqrt{n}\log^3n)$ gates. We can implement this using unitary gates (as explained in Appendix \ref{sec:classicaltoquantum}) and just use $O(\sqrt{n}\log^3n) = \widetilde{O}(\sqrt{n})$ ancilla qubits for all intermediate values in the circuit. Then we copy the final output $\ket{\prod_{i = 1}^d a_i^{t_i}}$ to a fresh register; each $a_i$ is $O(\log d)$ bits so we need $O(d \log d) = \widetilde{O}(\sqrt{n})$ ancilla qubits and $O(d \log d) = O(\sqrt{n} \log n)$ gates to do this. Finally, we can just uncompute the original circuit to restore all ancilla qubits to $\ket{0}$. In total, we have used $\widetilde{O}(\sqrt{n})$ ancilla qubits and $O(\sqrt{n}\log^3n)$ gates, as desired.
\condqed \end{proof}

Combining these ideas allows us to outline the algorithm. For each $j \in [K]$, define $c_j = \prod_{i = 1}^d a_i^{z_{i, j}}$. Then we have:
\begin{align*}
    \prod_{i = 1}^d a_i^{z_i + D/2} &= \prod_{j = 1}^K \left(\prod_{i = 1}^d a_i^{z_{i, j}}\right)^{F_j} 
    = \prod_{j = 1}^K c_j^{F_j}.
\end{align*}
We can use Lemma \ref{lemma:regevmultiplication} to compute each $c_j$, and Lemma \ref{lemma:fibonacci} to finally compute this product. The only missing detail is that of computing $c_j^{-1}$, which we do as follows:
\begin{align*}
    c_j^{-1} &= \prod_{i = 1}^d a_i^{-z_{i, j}} 
    = \left(\prod_{i = 1}^d a_i^{-1}\right) \cdot \prod_{i = 1}^d a_i^{1-z_{i, j}}.
\end{align*}
We can compute the latter expression using Lemma \ref{lemma:regevmultiplication} since the exponents are now in $\left\{0, 1\right\}$. Then we can apply Lemma \ref{lemma:beauregardeasy} using the classical constant $\prod_{i = 1}^d a_i^{-1}$ to finish computing $c_j^{-1}$. We formally detail all of this in the next section.

\subsection{Proof of Lemma \ref{lemma:oracle}}

Our procedure is detailed in Algorithm \ref{algo:fulloracle}. We track the use of ancilla qubits in the algorithm description.

\begin{algorithm}
    \KwData{Initial state $\ket{z}$ and $S + (\alpha-A+6+o(1))n$ ancilla qubits in the $\ket{0}$ state.}
    \KwResult{Final state comprising $\ket{\prod_{i = 1}^d a_i^{z_i+D/2} \bmod N}$ and $S + (\alpha+5+o(1))n$ qubits in some state (which may not be $\ket{0}$ and could depend on $z$).}
    \begin{enumerate}
        \item Use Lemma \ref{lemma:zeckendorf} to compute and store the values $\ket{z_{i, j}}$ for all $i \in [d]$ and $j \in [K]$. Note that this step also ``overwrites'' the qubits storing $\ket{z}$. (This consumes $dK - d\log D = (\alpha - A + o(1))n$ qubits, leaving $S+(6+o(1))n$ ancilla qubits.)
        \item Set aside $4n$ ancilla qubits. These will store our states $\ket{\psi(x_1)}$ and $\ket{\psi(x_2)}$. Initialize $x_1 \leftarrow 1$ and $x_2 \leftarrow 1$. (This leaves $S+(2+o(1))n$ ancilla qubits.)
        \item Repeat the following for $j = K, K-1, \ldots, 1$ in that order:
        \begin{enumerate}
            \item Consider the qubits $z_{i, j'}$ for $i \in [d]$ and $j' \neq j$. There are $d(K-1) \geq (\alpha - o(1))n \geq (\frac{2}{\log\phi} - o(1))n > 2n$ such qubits, and we will not use any of them for this iteration of the loop. Hence we may take $n-1$ of these qubits and pre-pend one clean qubit in the $\ket{0}$ state to obtain $n$ dirty ancilla qubits. We now have $S+2n$ clean ancilla qubits available.
            \item Update $x_1 \leftarrow x_1x_2$ using Lemma \ref{lemma:beauregardmult}. (This temporarily uses and restores $S$ clean ancilla qubits and $n$ dirty ancilla qubits, which we have.)
            \item Now we prepare the state $\ket{\psi(c_j)}$:
            \begin{enumerate}
                \item Calculate the state $\ket{\prod_{i = 1}^d a_i^{1-z_{i, j}} \bmod N}$ using Lemma \ref{lemma:regevmultiplication} and store it in another $n$-bit register. (We now have $S+n$ clean ancilla qubits.)
                \item Use Lemma \ref{lemma:beauregardeasy} with the classical constant $\prod_{i = 1}^d a_i^{-1} \bmod N$ to update the register from the above step to contain $\ket{c_j^{-1} \bmod N}$. (This will use and then restore all $S+n$ of our remaining clean ancilla qubits, as well as modifying our $n$ dirty ancilla qubits.)
                \item Calculate the state $\ket{c_j}$ using Lemma \ref{lemma:regevmultiplication} and store it in an $n$-bit register. (We now have $S$ clean ancilla qubits.) 
            \end{enumerate}
                
            \item We now have the state $\ket{\psi(c_j)}$, so we can update $x_1 \leftarrow x_1c_j$ using Lemma \ref{lemma:beauregardmult}. (This will use and then restore $S$ clean ancilla qubits and $n$ dirty ancilla qubits.)
            \item Now we uncompute the state $\ket{\psi(c_j)}$, returning all qubits to $\ket{0}$. This can be done by just applying in reverse order the inverses of the circuits used to construct this state. (This will use $S$ of our available clean ancilla qubits and then also free up the $2n$ ancilla qubits containing $\ket{\psi(c_j)}$, so we are now back to $S+2n$ clean ancilla qubits. Additionally, this will return all dirty ancilla qubits to their original value, since the dirty ancilla qubits are only modified in the construction of $\ket{\psi(c_j)}$ when we use Lemma \ref{lemma:beauregardeasy}.)
            \item Swap $x_1$ and $x_2$ (i.e. swap the registers $\ket{\psi(x_1)}$ and $\ket{\psi(x_2)}$).
        \end{enumerate}
    \end{enumerate}
    \caption{Quantum oracle for $\prod_{i = 1}^d a_i^{z_i+D/2}\bmod N$}\label{algo:fulloracle}
\end{algorithm}
First we address the correctness of our algorithm. We make some important observations to justify our use of dirty ancilla qubits:
\begin{itemize}
    \item At step $j$ in the loop, our algorithm only uses qubits $z_{i, j'}$ for $j' \neq j$ as dirty ancilla qubits, so the fact that these qubits might change value during the loop will not affect how $x_1$ and $x_2$ are updated.
    \item We also need to check that our dirty ancilla qubits are usable in Lemmas \ref{lemma:beauregardeasy} and \ref{lemma:beauregardmult} i.e. that $g \in [0, N-1]$, where $\ket{g}$ is the value stored by these ancilla qubits. Initially, this is ensured by the fact that the first bit of $g$ is 0 i.e. $g < 2^{n-1} \leq N$. Additionally, even though Lemma \ref{lemma:beauregardeasy} modifies $g$, it replaces $g$ with some other value that is reduced mod $N$. So $g$ will always be in $[0, N-1]$.
    \item Any application of Lemma \ref{lemma:beauregardmult} will preserve the value of these dirty ancilla qubits.
    \item The one application of Lemma \ref{lemma:beauregardeasy} when constructing $\ket{\psi(c_j)}$ will modify the dirty ancilla qubits, but this will be reversed at the end of step $j$ in the loop when uncomputing $\ket{\psi(c_j)}$. Hence the value of these qubits is ultimately restored at the end of the loop.
\end{itemize}

Now that we have shown that our use of dirty ancilla qubits does not affect anything, correctness follows from the proof of Lemma \ref{lemma:fibonacci}; at the end of step 3, we will have $\ket{x_2} = \ket{\prod_{j = 1}^K c_j^{F_j} \bmod N} = \ket{\prod_{i = 1}^d a_i^{z_i+D/2} \bmod N}$.

Finally, we check the size of our circuit. The steps before the loop use $O(n^{3/2})$ gates. Now, for each step of the loop over $j$:
\begin{enumerate}
    \item Each application of Lemma \ref{lemma:beauregardmult} uses $O(G+n)$ gates.
    \item Computing and later uncomputing the state $\ket{\psi(c_j)}$ uses $O(G+n)$ gates.
    \item The final swap uses $O(n)$ gates.
\end{enumerate}
The loop runs for $K$ iterations, so the total circuit size from the loop will be $O(K(G+n)) = O(n^{1/2}(G+n))$ gates. Combining this with the initialization steps yields a circuit of size $O(n^{1/2} \cdot G + n^{3/2})$, thus completing the proof of Lemma \ref{lemma:oracle}. \qed 

\subsection{Potential Future Space Optimizations}\label{sec:furtheropt}

We have shown that Regev's factoring algorithm can be implemented using just
$$S + \bigg(\frac{A}{\log\phi} + 6 + o(1)\bigg)n$$
qubits, while retaining the asymptotic $O(n^{1/2} \cdot G + n^{3/2})$ circuit size. We remind the reader that here $G$ and $S$ denote the number of gates and the number of ancilla qubits respectively for our multiplication circuit on $n$-bit integers.

A natural question is whether this can be optimized further, particularly the constant $\frac{A}{\log\phi} + 6$ in the linear term. To this end, we describe the various space costs incurred by our algorithm below:

\begin{enumerate}
    \item $\alpha n = \big(\frac{A}{\log\phi} + o(1)\big)n$ qubits are used to store the qubits $z_{i, j}$ for $i \in [d]$ and $j \in [K]$. This is a slight blow-up from the $(A+o(1))n$ bits that we really need to store all the $z_i$'s.
    
    This does appear inefficient, because at any given point in our algorithm we are only interested in one particular value of $j$. However, the natural way to compute the $z_{i, j}$'s involves iterating through each value of $i$ (since we decompose $z_i + D/2$ as a sum of Fibonacci numbers to obtain the $z_{i, j}$'s). A potential workaround would be to find a way to compute $z_{i, j}$ given $i$ and $j$ in a more efficient way than actually calculating the entire representation of $z_i + D/2$ as a sum of Fibonacci numbers, but we are not aware of how to do this.

    In follow-up work~\cite{cryptoeprint:2024/636}, we take a different approach to show that this space cost can be driven down to $(A+\epsilon)n$ for any $\epsilon > 0$. The idea is to work with more general sequences of integers than the Fibonacci numbers that yield encodings $z_{i, j}$ that are less redundant than the Zeckendorf representation.

    \item $S+n$ clean ancilla qubits are needed for the multiplication operations given in Lemmas \ref{lemma:beauregardeasy} and \ref{lemma:beauregardmult}. This is already quite optimized due to our use of dirty ancilla qubits in these lemmas. However, under certain conditions this can be brought down to $S$ (which is necessary for us to be able to use our multiplication circuit).
    
    The only reason we need the extra $n$ ancilla qubits is in Lemma \ref{lemma:beauregardeasy}, which we use with $a = (\prod_{i = 1}^d a_i^{-1}) \bmod N$. These $n$ qubits are used to write in $a$ and $-a^{-1}$ so that we can use our multiplication circuit to multiply by these integers. Note however that $a$ and $-a^{-1}$ are classical constants and can hence be precomputed classically.

    If the quantum multiplication circuit is just an implementation of a classical multiplication circuit with unitaries (e.g. when we use the $O(n \log n)$ multiplication algorithm by~\cite{Harvey21}), then we can save $n$ clean ancilla qubits as follows. Following ideas by Shor~\cite{shor97}, one can bake the values $a$ and $-a^{-1}$ into the given classical circuit architecture to obtain hardwired classical circuits for ``multiplication by $a$'' and ``multiplication by $-a^{-1}$''. Then these hardwired classical circuits can be implemented with unitary gates. Now we no longer require $a$ and $-a^{-1}$ to be explicitly written down in another register.

    For general quantum multiplication circuits, we are not aware of a way to achieve this hardwiring.
    

    \item $4n$ ancilla qubits are needed to store the ``accumulator'' states $\ket{\psi(x_1)}$ and $\ket{\psi(x_2)}$. We are not aware of any approaches to bypass this requirement.



    \item $2n$ ancilla qubits are used to store the state $\ket{\psi(c_j)} = \ket{c_j}\ket{c_j^{-1}\bmod N}$, but $n$ of these qubits are reused from the aforementioned $S+n$ clean ancilla qubits. This is still potentially unnecessary; $c_j = \prod_{i = 1}^d a_i^{z_{i, j}}$ is a product of at most $d$ integers of $O(\log d)$ bits each, so $c_j$ only comprises $O(d \log d) = \widetilde{O}(\sqrt{n})$ bits. While the same cannot be said of $c_j^{-1} \bmod N$, it might be possible to instead store $\prod_{i = 1}^d a_i^{1-z_{i, j}}$ (which also only comprises $\widetilde{O}(\sqrt{n})$ bits) in its place, and bake the fixed multiplier $(\prod_{i = 1}^d a_i^{-1}) \bmod{N}$ into the circuit architecture.

    Once again, we used $n$ qubits for each of these registers so that we could maintain a clean abstraction for multiplication operations, namely multiplying two $n$-bit integers that are reduced mod $N$. It may be possible to improve this if the multiplication algorithm being used can multiply a small integer with an $n$-bit integer without padding to make both integers $n$ bits.
    
\end{enumerate}

\section{Error-Resilience for Regev's Classical Postprocessing}\label{sec:error-resilience}

In this section, we prove Lemma \ref{lemma:ourpostprocessing}, our main error correction lemma for Regev's postprocessing algorithm. We start by introducing our error detection and filtering algorithm. The intuition behind it is that uncorrupted samples will be very close to $\Lambda^\ast/\ZZ^d$, so it should be easy to find small integer relations between them. Moreover, it should be much harder to find a small integer relation between our uncorrupted samples and even one corrupted sample. We formalize this as a shortest-vector problem in a suitably constructed lattice, and will see that the approximate solution to this via LLL~\cite{lenstra1982factoring} is sufficient for our purposes. Using this idea, we can identify sufficiently many uncorrupted samples and apply Regev's postprocessing procedure~\cite{Regev23} to these. Our algorithm is described in Algorithm \ref{algo:allpostprocessing}.

Our analysis proceeds in three steps:

\begin{enumerate}
    \item In Section \ref{sec:completeness}, we show that with probability $1-o(1)$, Algorithm \ref{algo:allpostprocessing} will add at least 1 sample to $B$ in each iteration until $|B| \geq \gamma d$, hence proving that it terminates in polynomial time.
    \item In Section \ref{sec:corruptfiltering}, we formalize the above intuition and argue that none of the samples added to $B$ will be corrupted, with probability $1-o(1)$.
    \item Although we use Regev's postprocessing procedure as-is, it turns out that Regev's analysis is not directly applicable to our algorithm. In Section \ref{sec:regevanalysis}, we outline the necessary modifications and elaborate in Appendix \ref{sec:regevlatticeanalysis}.
\end{enumerate}

\begin{algorithm}
    \KwData{Norm bound $T$ and noisy samples $w_1, w_2, \ldots, w_{m} \in \RR^d/\ZZ^d$ from $\Lambda^\ast/\ZZ^d$ corrupted according to the model in Section \ref{sec:errormodel}.}
    \KwResult{A finite list of elements of $\Lambda$, such that they generate any element $u \in \Lambda$ with $||u||_2 \leq T$.}
    Initialize $B = \emptyset$. $B$ will track the subset of samples that we select. Let $S = 2^{An/d}$. Recall that $m = \alpha d$.
    \begin{enumerate}
        \item Repeat the following until $|B| \geq \gamma d$:
        \begin{enumerate}
            \item Let $E$ be a subset of $[m]$ such that $E \cap B = \emptyset$ and $|E| = (\alpha-\gamma)d$. This exists because $|B| \leq \gamma d$.
            \item Define the matrix $W \in \RR^{d \times |E|}$ to be the matrix with columns $w_i$ for $i \in E$.
            \item Define the lattice $\Lambda'' \subseteq \RR^{d+|E|}$ as that spanned by the columns of $$H = \begin{pmatrix} SI_{d \times d} &\rvline & \begin{matrix} SW \end{matrix} \\ \hline 0 &\rvline & I_{|E| \times |E|}\end{pmatrix}.$$
            We note that this is similar but not identical to the lattice construction used in Regev's postprocessing procedure~\cite{Regev23}.
            \item Apply LLL basis reduction~\cite{lenstra1982factoring} to $\Lambda''$ to find a $2^{(d+|E|-1)/2}$-approximate shortest vector in $\Lambda''$.
            \item This vector can be written as $$H\begin{pmatrix} \beta \\ \hline a_i: i \in E \end{pmatrix} = \begin{pmatrix} S(\beta + \sum_{i \in E} a_iw_i) \\ \hline a_i: i \in E\end{pmatrix},$$where $\beta \in \ZZ^d$ is some vector and $a_i: i \in E$ denotes a column vector in $\ZZ^{|E|}$.
            \item For each $i \in E$, add $i$ to $B$ if and only if $a_i \neq 0$.
        \end{enumerate}
        \item Let $B' \subseteq B$ be arbitrary with $|B'| = \gamma d$. Apply Regev's classical postprocessing (Algorithm \ref{algo:regevpostprocessing}) to the collection $\left\{w_i: i \in B'\right\}$.
    \end{enumerate}
    \caption{Classical postprocessing of the outputs from Regev's quantum circuit}\label{algo:allpostprocessing}
\end{algorithm}

\subsection{Proof that Our Algorithm Selects $\gamma d$ Samples}\label{sec:completeness}

Firstly, note that since $p < 1 - (\gamma+1)/\alpha$, with probability $1-o(1)$ at most $(\alpha-\gamma-1)d$ of our samples from the quantum circuit will be corrupted. We now show an upper bound on the shortest vector of $\Lambda''$ in the next two lemmas:

\begin{lemma}\label{lemma:regevPHP}
    Given vectors $v_1, \ldots, v_d \in \Lambda^\ast$, there exist $a_1, a_2, \ldots, a_d \in \ZZ$, not all zero, such that $|a_i| \leq 2^{n/d}$ for all $i$ and $\sum_i a_iv_i \in \ZZ^d$.
\end{lemma}
\begin{proof}
    This follows from a pigeonhole principle argument similar to that observed by~\cite{Regev23}, which we restate here. Consider the sums $\sum_{i = 1}^d a_iv_i \in \Lambda^\ast$ for all $|a_i| \leq 2^{n/d-1}$. There are $(2^{n/d} + 1)^d > 2^n > \det\Lambda = |\Lambda^\ast/\ZZ^d|$ such sums, so some two belong in the same coset of $\Lambda^\ast/\ZZ^d$. The difference between those sums will hence lie in $\ZZ^d$, and each coefficient will have magnitude $\leq 2^{n/d}$. At least one coefficient will be nonzero.
\condqed \end{proof}

\begin{lemma}\label{lemma:existsshortvector}
    With probability $1-o(1)$, for any $E$ with $|E| = (\alpha-\gamma)d$
    there exists a nonzero vector in the corresponding $\Lambda''$ of length at most $2^{(1+o(1))n/d}$.
\end{lemma}
\textit{Proof.}
    Observe firstly that if $|E| = (\alpha-\gamma) d$ and we know that with probability $1-o(1)$ we have at most $(\alpha-\gamma-1)d$ corrupted samples in total, this implies that there are $\geq d$ uncorrupted samples in $E$. By Lemma \ref{lemma:regevPHP}, it follows that there exist integers $a_i$ for $i \in E$ such that:
    \begin{itemize}
        \item $\sum_{i \in E} a_iv_i \in \ZZ^d$.
        \item $|a_i| \leq 2^{n/d}$ for all $i$.
        \item $a_i$ is nonzero for at most $d$ values of $i$, and these indices correspond to uncorrupted samples.
    \end{itemize}
    Hence we can set $\beta = -\sum_{i \in E} a_iv_i \in \ZZ^d$. Then consider the lattice element $$H\begin{pmatrix} \beta \\ \hline a_i: i \in E\end{pmatrix} = \begin{pmatrix} S(\beta + \sum_{i \in E} a_iw_i) \\ \hline a_i: i \in E\end{pmatrix}.$$This is nonzero and an element of $\Lambda''$. Moreover, we have:
    \begin{align*}
        \beta + \sum_{i \in E} a_iw_i &= \beta + \sum_{i \in E} a_iv_i + \sum_{i \in E} a_i\delta_i 
        = \sum_{i \in E} a_i\delta_i,
    \end{align*}
    which has norm at most $m 2^{n/d} \cdot 2^{(-A+o(1))n/d} = m2^{(-A+1+o(1))n/d}.$ (Note that $a_i$ is 0 for any $i$ corresponding to a corrupted sample). Finally, we have $\sum_{i \in E} a_i^2 \leq d \cdot 2^{2n/d}.$ It follows that the norm of this element in $\Lambda''$ is at most:
    \begin{align*}
        \sqrt{S^2m^22^{(-2A+2+o(1))n/d} + d \cdot 2^{2n/d}} &= \sqrt{m^22^{(2+o(1))n/d} + d \cdot 2^{2n/d}} \\
        &\leq \sqrt{(m^2+d)2^{(2+o(1))n/d}} \\
        &= 2^{(1+o(1))n/d}. \hspace{1.4in} \condqed
    \end{align*}

\begin{lemma}\label{lemma:gainonesample}
    With probability $1-o(1)$, as long as $|B| < \gamma d$, our algorithm will add at least one sample to $B$ in each iteration.
\end{lemma}
\begin{proof}
    By Lemma \ref{lemma:existsshortvector}, our algorithm will find a nonzero vector in $\Lambda''$ of $\ell_2$ norm at most $$2^{(d+|E|-1)/2} \cdot 2^{(1+o(1))n/d} \leq 2^{(\alpha-\gamma+3+o(1))n/(2d)}.$$
    We can write this vector as$$H\begin{pmatrix} \beta \\ \hline a_i: i \in E\end{pmatrix} = \begin{pmatrix} S(\beta + \sum_{i \in E} a_iw_i) \\ \hline a_i: i \in E\end{pmatrix}.$$Suppose for the sake of contradiction that our algorithm adds no samples to $B$ in this round. Then $a_i$ must be 0 for all $i \in E$. It follows that $\beta + \sum_{i \in E} a_iw_i = \beta$ must be a nonzero vector in $\ZZ^d$. It follows that the norm of this vector is at least $S||\beta||_2 \geq S = 2^{An/d}$, which is a contradiction for $A \geq \frac{\alpha-\gamma+3}{2} + o(1)$.
\condqed \end{proof}
It follows that in each round of the algorithm, $|B|$ will increase by at least 1. This implies that Algorithm \ref{algo:allpostprocessing} will terminate in polynomial time, and it will find a subset of $\geq \gamma d$ samples to pass to Regev's postprocessing procedure~\cite{Regev23}.

\subsection{Proof that Our Selected Samples are Not Corrupted}\label{sec:corruptfiltering}

Our next step is to show that our algorithm will not select any corrupted samples to be passed to Regev's postprocessing procedure. We address this in the next two lemmas.

\begin{lemma}\label{lemma:nobadlincombs}
    With probability $1-o(1)$, there do not exist integers $a_1, a_2, \ldots, a_m$ such that:
    \begin{itemize}
        \item $|a_i| \leq 2^{(\alpha-\gamma+3+o(1))n/(2d)}$ for all $i$.
        \item $a_i \neq 0$ for at least one index $i$ corresponding to a corrupted sample.
        \item $d(\sum_{i\text{ corrupted}} a_i\epsilon_i, \Lambda^\ast/\ZZ^d) \leq 2^{(-2A+\alpha-\gamma+3+o(1))n/(2d)}.$
    \end{itemize}
\end{lemma}
\begin{proof}
    We can argue by conditioning on the subset of our $m$ samples that are corrupted. The terms from cases where there are $> (\alpha - \gamma - 1)d$ corrupted samples contribute a probability of $o(1)$. When there are $\leq (\alpha - \gamma - 1)d$ corrupted samples, we can directly use the assumption that $\mathcal{D}$ (the distribution of each $\epsilon_i$) is well-spread, with $k \leq (\alpha - \gamma - 1)d$ set to the number of corrupted samples.
\condqed \end{proof}

\begin{lemma}\label{lemma:nobadsamples}
    With probability $1-o(1)$, as long as $|B| < \gamma d$, our algorithm will not add any corrupted samples to $B$ in any iteration.
\end{lemma}
\begin{proof}
    As in the proof of Lemma \ref{lemma:gainonesample}, it follows from Lemma \ref{lemma:existsshortvector} that our algorithm will find a nonzero vector in $\Lambda''$ of $\ell_2$ norm at most $$2^{(d+|E|-1)/2} \cdot 2^{(1+o(1))n/d} \leq 2^{(\alpha-\gamma+3+o(1))n/(2d)},$$and moreover we can write this vector as$$H\begin{pmatrix} \beta \\ \hline a_i: i \in E\end{pmatrix} = \begin{pmatrix} S(\beta + \sum_{i \in E} a_iw_i) \\ \hline a_i: i \in E\end{pmatrix}.$$
    Both of the following must hence hold:
    \begin{itemize}
        \item $|a_i| \leq 2^{(\alpha-\gamma+3+o(1))n/(2d)}$ for all $i \in E$.
        \item $||\beta + \sum_{i \in E} a_iw_i||_2 \leq S^{-1} 2^{(\alpha-\gamma+3+o(1))n/(2d)} = 2^{(-2A+\alpha-\gamma+3+o(1))n/(2d)}.$
    \end{itemize}
    Letting $\mathcal{C}$ denote the set of corrupted samples, we can write the latter of these two expressions as follows:
    \begin{align*}
        & \beta  + \sum_{i \in E} a_iw_i 
         = \beta + \sum_{i \in E\cap\overline{\mathcal{C}}} a_i(v_i+\delta_i) + \sum_{i \in E\cap \mathcal{C}} a_i(\eta(v_i+\delta_i)+\epsilon_i) \\
        &= \bigg(\beta + \sum_{i \in E\cap \overline{\mathcal{C}}}a_iv_i + \sum_{i \in E\cap \mathcal{C}}\eta a_iv_i\bigg) 
        + \bigg(\sum_{i \in E\cap \overline{\mathcal{C}}} a_i\delta_i + \sum_{i \in E\cap \mathcal{C}}\eta a_i\delta_i\bigg) + \bigg(\sum_{i \in E\cap \mathcal{C}} a_i\epsilon_i\bigg)
    \end{align*}
    Rearranging, we get 
    \begin{align*}
        \sum_{i \in E\cap \mathcal{C}} a_i\epsilon_i &= \bigg(\beta + \sum_{i \in E} a_iw_i\bigg) - \bigg(\sum_{i \in E\cap\overline{\mathcal{C}}} a_i\delta_i + \sum_{i \in E\cap \mathcal{C}}\eta a_i\delta_i\bigg) \\
        &-\bigg(\beta + \sum_{i \in E\cap \overline{\mathcal{C}}}a_iv_i + \sum_{i \in E\cap \mathcal{C}}\eta a_iv_i\bigg)
    \end{align*}
    The last of these three terms is some element in $\Lambda^\ast/\ZZ^d$, noting that $\eta$ and the $a_i$'s are all integers and that $\beta \in \ZZ^d$. We can bound the second term as follows:
    \begin{align*}
        \norm*{\sum_{i \in E\cap \overline{\mathcal{C}}} a_i\delta_i + \sum_{i \in E\cap \mathcal{C}}\eta a_i\delta_i}_2 &\leq \alpha d \cdot \max_i |a_i| \cdot \max_i ||\delta_i||_2 \\
        &\leq \alpha d \cdot 2^{(\alpha-\gamma+3+o(1))n/(2d)} \cdot 2^{(-A+o(1))n/d} \\
        &= 2^{(-2A+\alpha-\gamma+3+o(1))n/(2d)}.
    \end{align*}
    Finally, we already observed that the first term also has magnitude at most $2^{(-2A+\alpha-\gamma+3+o(1))n/(2d)}$. It follows by the triangle inequality that
    $$d\bigg(\sum_{i \in E\cap \mathcal{C}}a_i\epsilon_i, \Lambda^\ast/\ZZ^d\bigg) \leq 2^{(-2A+\alpha-\gamma+3+o(1))n/(2d)}.$$
    By Lemma \ref{lemma:nobadlincombs}, this forces $a_i = 0$ for all $i \in E$ such that sample $i$ is corrupted. This implies that our algorithm will not add any corrupted samples to $B$.
\condqed \end{proof}

\subsection{Proof that Regev's Postprocessing Works with Any $\gamma d$ Uncorrupted Samples}\label{sec:regevanalysis}

Regev's postprocessing argument~\cite{Regev23} works with $d+4$ independent and uncorrupted samples from $\Lambda^\ast/\ZZ^d$ and proceeds as follows:
\begin{enumerate}
    \item First, he applies a result due to Pomerance~\cite{pomerance2002expected} to argue that these $d+4$ elements will generate $\Lambda^\ast/\ZZ^d$ with probability at least $1/2$.
    \item He then argues that it is unlikely for there to exist $u \in \ZZ^d\backslash\Lambda$ which is nearly orthogonal to all $d+4$ samples. Thus these samples can be used to construct a suitable lattice such that applying LLL basis reduction~\cite{lenstra1982factoring} yields a basis for short vectors in $\Lambda$.
\end{enumerate}
In our case, we do not quite have independence; instead, what we have is a subset of size $\gamma d$ out of $\alpha d$ independent samples from $\Lambda^\ast/\ZZ^d$. Our analysis in Section \ref{sec:corruptfiltering} implies that we can assume these samples are uncorrupted. Informally, we just need to argue that this subset is ``close enough'' to independent that all of Regev's analysis still works. This largely mirrors the existing analysis in~\cite{Regev23}, so we defer details to Appendix \ref{sec:regevlatticeanalysis}.

\ifanon
\else
\condparagraph{Acknowledgements.} We thank Oded Regev for suggesting the noise-robustness problem to us and for outlining an approach that led to our solution.  We thank Oded, Martin Eker{\aa}, Joel G{\"a}rtner, and anonymous reviewers for their detailed comments on this manuscript, and Madhu Sudan for discussions about the noise-robustness problem. SR would also like to thank Ittai Rubinstein for helpful discussions about lattices and the noise-robustness problem, and Gregory D. Kahanamoku-Meyer and Katherine van Kirk for feedback on Table~\ref{summarytable} and numerous discussions related to Regev's algorithm and quantum multiplication circuits. VV would like to thank Muli Safra, Subhash Khot, Dor Minzer and the Simons Foundation for organizing a stimulating workshop on lattices, and Oded Regev for his lecture at the workshop, which initiated our thinking on the problems in this paper. SR was supported by an Akamai Presidential Fellowship, NSF CNS-2154149, and a Simons Investigator Award. VV was supported in part by DARPA under Agreement Number HR00112020023, NSF CNS-2154149, a grant from the MIT-IBM Watson AI, a Thornton Family Faculty Research Innovation Fellowship from MIT and a Simons Investigator Award. Any opinions, findings and conclusions or recommendations expressed in this material are those of the author(s) and do not necessarily reflect the views of the United States Government or DARPA.
\fi

\bibliographystyle{alpha}
\bibliography{main}

\appendix

\section{Implementation of Our Multiplication Oracle}\label{sec:multimplementation}

In Theorem \ref{thm:mainresult} and our algorithm, we assume a quantum multiplication circuit that takes a specific form, mapping $$\ket{a}\ket{b}\ket{t}\ket{0^S} \mapsto \ket{a}\ket{b}\ket{(t+ab)\bmod N}\ket{0^S}~.$$
Existing multiplication circuits such as those by~\cite{Harvey21} and~\cite{gidney2019} may differ from this form in three ways:
\begin{itemize}
    \item Circuits such as that by~\cite{Harvey21} may be classical circuits rather than quantum.
    \item These circuits, as in both~\cite{Harvey21} and~\cite{gidney2019}, may only support multiplication over $\ZZ$ rather than modulo $N$ i.e. they produce $ab$ as an output rather than $ab\bmod{N}$.
    \item Even if these circuits could compute $ab\bmod{N}$, they may not have the specific structure mapping $\ket{t}$ to $\ket{(t+ab)\bmod{N}}$ that we require.
\end{itemize}
In this section, we show how to overcome each of these gaps:
\begin{itemize}
    \item In Section \ref{sec:classicaltoquantum}, we restate well-known procedures~\cite{bennett_logical_1973, bennett_timespace_1989, levine_note_1990} for ``compiling'' a classical circuit into a quantum circuit that can compute the same function out of place.
    \item In Section \ref{sec:modNoracle}, we show how to go from computing $ab$ to computing $ab\bmod{N}$.
    \item In Section \ref{sec:ringstructuremodN}, we show how to achieve the specific structure computing $(t+ab) \bmod{N}$ that we need.
\end{itemize}

Additionally, we quickly verify in Section \ref{sec:checkharveyandgidney} that these arguments yield the space and size guarantees we claim when using the multiplication circuits in~\cite{Harvey21} and~\cite{gidney2019}. Finally, we present a different oracle construction for the special case of schoolbook multiplication in Section \ref{sec:schoolbookoracle}. This avoids the space overheads associated with the above ``conversion'' procedures, and only requires $O(1)$ ancilla qubits.

\subsection{Compiling Classical Circuits into Quantum Circuits}\label{sec:classicaltoquantum}

In this section, we state and reprove the well-known fact~\cite{bennett_logical_1973, bennett_timespace_1989, levine_note_1990} that a classical circuit can be ``compiled'' into a quantum circuit that can carry out the same computation in superposition, while tracking the space and size of such a quantum circuit.

\begin{lemma}
    (Well-known~\cite{bennett_logical_1973, bennett_timespace_1989, levine_note_1990}) Assume there exists a classical circuit with $G$ wires (including input and output bits) that computes some function $f(x) \in \left\{0, 1\right\}^{n_o}$ of its input $x \in \left\{0, 1\right\}^{n_i}$. Then there exists a quantum circuit using $O(G)$ gates mapping $$\ket{x}\ket{0^{G+n_o-n_i}} \mapsto \ket{x} \ket{0^{G-n_i}} \ket{f(x)}.$$
    Note that the number of qubits in such a circuit is $G+n_o \leq 2G$ (since $G$ includes output wires).
\end{lemma}
\begin{proof}
    Assume without loss of generality that the classical circuit comprises only NOT and AND gates.
    
    Let us temporarily ignore $n_o$ of the ancilla qubits. The remaining ancilla qubits can be assigned bijectively to non-input wires in the classical circuit. Using these qubits, we simulate the classical circuit as follows:
    \begin{itemize}
        \item For a NOT gate with input wire corresponding to qubit $\ket{a}$ and output wire corresponding to qubit $\ket{b}$, we can apply a CNOT gate with control qubit $\ket{a}$ and target qubit $\ket{b}$, and then apply an X gate to $\ket{b}$.
        \item For an AND gate with input wires corresponding to qubits $\ket{a}, \ket{b}$ and output wire corresponding to qubit $\ket{c}$, we can apply a CCNOT gate with control qubits $\ket{a}, \ket{b}$ and target qubit $\ket{c}$.
    \end{itemize}
    At the end of this computation, we will have the state $$\ket{x} \ket{\phi} \ket{f(x)} \ket{0^{n_o}},$$ for some state $\phi$ on $G - n_i - n_o$ qubits. We can then copy $\ket{f(x)}$ to the $n_o$ remaining ancilla qubits using CNOT gates, and then uncompute the classical circuit simulation to obtain the final desired state. The number of gates in our quantum circuit is clearly $O(G)$.
\condqed \end{proof}

\subsection{Computing $ab \bmod{N}$}\label{sec:modNoracle}

To go from computing $ab$ to computing $ab \bmod{N}$, the main ingredient is unsurprisingly an efficient quantum circuit for out-of-place reduction modulo $N$:

\begin{lemma}\label{lemma:reducemodN}
    Assume there exists a quantum circuit mapping $$\ket{a} \ket{b} \ket{0^{2n}} \ket{0^S} \mapsto \ket{a} \ket{b} \ket{ab} \ket{0^S}$$ with $G$ gates, where $a, b$ are $n$-bit integers.
    
    Let $c$ be a $2n$-bit integer. Then there exists a quantum circuit using $O(G+n)$ gates mapping $$\ket{c}\ket{0^{S+O(n)}} \mapsto \ket{c} \ket{c \bmod{N}} \ket{0^{S+O(n)}}.$$
\end{lemma}
\begin{proof}
    Let $b = 2n$ be a precision parameter (so $c < 2^b$). Assume we have classically precomputed a high-precision approximation of $1/N$; concretely, we have an integer $r$ so that $$r/2^b < 1/N < (r+1)/2^b.$$
    The idea is to use this to approximately divide $c$ by $N$, which we do as follows:
    \begin{align*}
        \ket{c} \ket{r} \ket{N} \ket{0^{S+O(n)}} &\rightarrow \ket{c} \ket{r} \ket{N} \ket{cr} \ket{0^{S+O(n)}} \\
        &= \ket{c} \ket{r} \ket{N} \ket{\lfloor cr/2^b \rfloor} \ket{cr\bmod{2^b}} \ket{0^{S+O(n)}} \\
        &~~~~\text{(splitting the register containing $cr$ into two pieces)} \\
        &\rightarrow \ket{c} \ket{r} \ket{N} \ket{\lfloor cr/2^b \rfloor} \ket{0^{S+O(n)}} \ket{cr\bmod{2^b}} \ket{N \cdot \lfloor cr/2^b \rfloor} \ket{0^{S+O(n)}} \\
        &\rightarrow \ket{c} \ket{r} \ket{N} \ket{\lfloor cr/2^b \rfloor} \ket{0^{S+O(n)}} \ket{cr\bmod{2^b}} \ket{c - N \cdot \lfloor cr/2^b \rfloor} \ket{0^{S+O(n)}},
    \end{align*}
    where the last step is achieved by using the inverse of an in-place integer addition circuit e.g. see~\cite{draper}. Next, observe that:
    \begin{align*}
        N \cdot \lfloor cr/2^b \rfloor &\leq Ncr/2^b \\
        &< c \\
        \Rightarrow 0 &< c-N \cdot \lfloor cr/2^b \rfloor,\text{ and } \\
        N \cdot \lfloor cr/2^b \rfloor &> N \cdot (cr/2^b - 1) \\
        &> N \cdot (c \cdot (1/N - 1/2^b) - 1) \\
        &> c - 2N \\
        \Rightarrow 2N &> c-N \cdot \lfloor cr/2^b \rfloor.
    \end{align*}
    So the number we have computed is congruent to $c$ modulo $N$, and is contained in the interval $(0, 2N)$. This implies that it is either equal to $c \bmod{N}$ or $(c \bmod{N}) + N$. These cases can be distinguished and $c\bmod{N}$ computed out of place using $O(n)$ qubits and $O(n)$ gates. Finally, this can be copied to a fresh register and the above computation uncomputed.

    In terms of circuit size, the only operation that is not $O(n)$ gates is multiplying pairs of integers of $\leq 2n$ bits a constant number of times. This can in turn be done using a constant number of calls to the given quantum circuit (one call does not suffice, since we are assuming a circuit that multiplies $n$-bit integers rather than $2n$-bit). The conclusion follows.
\condqed \end{proof}
Now, it is straightforward to check that we can compute $ab\bmod{N}$:
\begin{lemma}
    Using the same multiplication circuit as in Lemma \ref{lemma:reducemodN}, there exists a quantum circuit using $O(G+n)$ gates mapping $$\ket{a} \ket{b} \ket{0^n} \ket{0^{S+O(n)}} \mapsto \ket{a} \ket{b} \ket{ab\bmod{N}} \ket{0^{S+O(n)}}.$$
\end{lemma}
\begin{proof}
    We proceed as follows:
    \begin{align*}
        \ket{a} \ket{b} \ket{0^n} \ket{0^{S+O(n)}} &\mapsto \ket{a} \ket{b} \ket{ab} \ket{0^n} \ket{0^{S+O(n)}} \text{ (using the multiplication circuit)} \\
        &\mapsto \ket{a} \ket{b} \ket{ab} \ket{ab\bmod{N}} \ket{0^{S+O(n)}}, \text{ (using Lemma \ref{lemma:reducemodN})}
    \end{align*}
    after which we can copy and uncompute as usual to obtain the desired result.
\condqed \end{proof}

\subsection{Computing $(t+ab)\bmod{N}$}\label{sec:ringstructuremodN}

Now, we can assume we have a quantum circuit that maps $$\ket{a}\ket{b}\ket{0^n}\ket{0^{S_0}} \mapsto \ket{a}\ket{b}\ket{ab \bmod N}\ket{0^{S_0}}~.$$We now show how to obtain our specific functionality from this, at the expense of $n$ additional ancilla qubits. The key ingredients for our constructions in this section are the following space-efficient primitives for modular doubling and addition, due to \cite{roetteler17}:

\begin{lemma}\label{lemma:roetadd}
    {\em \cite{roetteler17}} There exists a circuit using $O(n\log n)$ gates computing the mapping 
    $$\ket{x}\ket{y}\ket{0}\ket{0} \mapsto \ket{x} \ket{(x+y)\bmod N} \ket{0}\ket{0}~.$$ 
    Here, $x$ and $y$ are reduced mod $N$, and each $\ket{0}$ is just one individual ancilla qubit.
\end{lemma}

\begin{lemma}\label{lemma:roetdbl}
    {\em \cite{roetteler17}} 
    Provided $N$ is odd, there exists a circuit using $O(n \log n)$ gates computing the mapping $$\ket{x}\ket{0}\ket{0} \mapsto \ket{2x \bmod N} \ket{0}\ket{0}~.$$ 
    Here, $x$ is once again reduced mod $N$, and we only have two ancilla qubits.
\end{lemma}

Using the addition lemma, we have our first simple result:
\begin{lemma}\label{lemma:modoracleconstruct}
    Suppose we have a quantum circuit using $G_0$ gates that maps $$\ket{a}\ket{b}\ket{0^n}\ket{0^{S_0}} \mapsto \ket{a}\ket{b}\ket{ab \bmod N}\ket{0^{S_0}},$$ and $S_0 \geq 2$. Then there is also a quantum circuit using $O(G_0 + n \log n)$ gates that maps $$\ket{a}\ket{b}\ket{t}\ket{0^{S_0+n}} \mapsto \ket{a}\ket{b}\ket{(t+ab)\bmod N}\ket{0^{S_0+n}},$$whenever $t$ is reduced mod $N$.
\end{lemma}
\begin{proof}
    We proceed as follows:
    \begin{align*}
        \ket{a} & \ket{b} \ket{t}\ket{0^n}\ket{0^{S_0}} \\ 
        &\rightarrow \ket{a}\ket{b}\ket{t}\ket{ab \bmod N}\ket{0^{S_0}}\text{ (using the given circuit)}\\
        &\rightarrow \ket{a}\ket{b}\ket{(t+ab)\bmod N}\ket{ab \bmod N} \ket{0^{S_0}}\text{ (using Lemma \ref{lemma:roetadd})} \\
        &\rightarrow \ket{a}\ket{b}\ket{(t+ab)\bmod N}\ket{0^n}\ket{0^{S_0}} \text{ (using the inverse of the given circuit)}
    \end{align*}
    The number of gates is clearly $O(G_0 + n\log n)$, so this establishes the lemma.
\condqed \end{proof}

\subsection{Special Cases of~\cite{Harvey21} and~\cite{gidney2019}}\label{sec:checkharveyandgidney}

Let us check that the compilation procedures presented here yield the space and size guarantees we claimed in Theorem \ref{thm:mainresult} when adapting the multiplication circuits by~\cite{Harvey21} and~\cite{gidney2019} to our situation.

In the case of~\cite{Harvey21}, we initially have a classical circuit of size $O(n \log n)$. Section \ref{sec:classicaltoquantum} converts this into a quantum circuit with size and space $O(n \log n)$ that computes $ab$. As for~\cite{gidney2019}, we initially have a quantum circuit with size $O(n^{\log_23})$ and space $O(n)$ for computing $ab$.

Now applying the conversions in Section \ref{sec:modNoracle} and Section \ref{sec:ringstructuremodN} yields circuits computing our desired $t \mapsto (t+ab)\bmod{N}$ functionality while retaining the asymptotic size and space guarantees; the size incurs a constant factor overhead plus $O(n\log n)$ extra gates, while the space incurs an $O(n)$ additive overhead.

\subsection{Special Case of Schoolbook Multiplication}\label{sec:schoolbookoracle}

Finally, we show how to construct a circuit using schoolbook multiplication where only $O(1)$ ancilla qubits are needed. \cite{roetteler17} already constructs a circuit mapping $\ket{a}\ket{b}\ket{0^n} \mapsto \ket{a}\ket{b}\ket{ab \bmod N}$, but we want to avoid the $O(n)$-qubit overhead of the conversion procedure in Section \ref{sec:ringstructuremodN}. Fortunately, this can be done with just a slight tweak of their shift-and-add construction:

\begin{lemma}\label{lemma:schoolbook}
    Provided $N$ is odd, there exists a quantum circuit using $O(n^2\log n)$ gates that maps $$\ket{a}\ket{b}\ket{t}\ket{0}\ket{0} \mapsto \ket{a}\ket{b}\ket{(t+ab)\bmod N}\ket{0}\ket{0},$$whenever $a, b, t$ are all reduced mod $N$.
\end{lemma}
\begin{proof}
    Label the bits comprising $a$ as $a_0, a_1, \ldots, a_{n-1}$, so that $a = \sum_{i = 0}^{n-1} a_i2^i$. Then observe that $$ab = \left(\sum_{i = 0}^{n-1} a_i2^i\right)b = \sum_{i = 0}^{n-1}a_i \cdot (2^ib).$$ This suggests the following algorithm:
    \begin{enumerate}
        \item Let $x, y, z$ denote the values in the register containing $a, b, t$ respectively. So initially we have $x = a, y = b, z = t$.
        \item Repeat the following for $i = 0, 1, \ldots, n-1$:
        \begin{enumerate}
            \item Use $a_i$ as a control bit for the circuit in Lemma \ref{lemma:roetadd} to update $z \leftarrow (z + a_iy) \bmod N$. We have the necessary two ancilla qubits for this.
            \item Use Lemma \ref{lemma:roetdbl} to update $y \leftarrow 2y \bmod N$. We have the necessary two ancilla qubits for this.
        \end{enumerate}
        \item Repeat the following $n$ times: update $y \leftarrow y/2 \bmod N$. This can be done using the inverse of the circuit from Lemma \ref{lemma:roetdbl}. Once again, this only needs two ancilla qubits.
    \end{enumerate}
    First, note that we make $O(n)$ calls to the circuits in Lemmas \ref{lemma:roetadd} and \ref{lemma:roetdbl}, implying the gate complexity of $O(n^2\log n)$.

    To see correctness, it is straightforward to show by induction that at the end of step $i$ within the loop in stage 2, we will have:
    \begin{align*}
        x &= a, \\
        y &= 2^{i+1}b \bmod N,\text{ and} \\
        z &= \left(t + \sum_{j = 0}^i a_j \cdot 2^j b\right) \bmod N.
    \end{align*}
    So at the end of stage 2, we will have $x = a$, $y = 2^nb \bmod N$, and $z = (t+ab)\bmod N$. Then after stage 3, $y$ will become $b$ again. The ancilla qubits are returned to 0 in each step, so the final state after applying our algorithm will be
    $$\ket{x}\ket{y}\ket{z}\ket{0}\ket{0} = \ket{a}\ket{b}\ket{(t+ab)\bmod N}\ket{0}\ket{0},$$as desired.
\condqed \end{proof}

\section{Analysis of Regev's Classical Postprocessing Procedure}\label{sec:regevlatticeanalysis}

In this section, we state Regev's classical postprocessing procedure and adapt the analysis by~\cite{Regev23} to work for our slightly different setting. As explained in Section \ref{sec:regevanalysis}, Regev's analysis assumes his postprocessing procedure is given $k = d+4$ independent noisy samples from $\Lambda^\ast/\ZZ^d$. Our setting is different because we have an arbitrary subset of size $k = \gamma d$ out of $m = \alpha d$ independent noisy samples from $\Lambda^\ast/\ZZ^d$. We first describe Regev's postprocessing procedure~\cite{Regev23} in Algorithm \ref{algo:regevpostprocessing}, then turn to its analysis.

\begin{algorithm}
    \KwData{Norm bound $T$ and a collection of noisy samples $w_1, w_2, \ldots, w_{k} \in \RR^d/\ZZ^d$ from $\Lambda^\ast/\ZZ^d$.}
    \KwResult{A finite list of elements of $\Lambda$, such that they generate any element $u \in \Lambda$ with $||u||_2 \leq T$.}
    \begin{enumerate}
        \item Define the lattice $\Lambda' \subseteq \RR^{d+k}$ as that spanned by the columns of $$\begin{pmatrix} I_{d \times d} &\rvline & 0 \\ \hline \begin{matrix} \delta^{-1}w_1 \\ \vdots \\ \delta^{-1}w_k \end{matrix} &\rvline & I_{k \times k}\end{pmatrix}.$$
        Recall that $\delta$ is an upper bound on the magnitude of additive noise $\delta_i$ included in each $w_i$.
        \item Let $z_1, \ldots, z_{d+k} \in \RR^{d+k}$ be an LLL reduced basis~\cite{lenstra1982factoring} of $\Lambda'$, and let $\tilde{z}_1, \ldots, \tilde{z}_{d+k}$ be the Gram-Schmidt orthogonalization of this basis.
        \item Let $l \geq 0$ be minimal such that $||\tilde{z}_{l+1}||_2 \geq 2^{(d+k)/2} \cdot (k+1)^{1/2} \cdot T$ (if no such $l$ exists, take $l = d+k$).
        \item For each $i \in [d+k]$, let $z'_i \in \RR^d$ consist of the first $d$ coordinates of $z_i$. Output $z'_1, \ldots, z'_l$.
    \end{enumerate}
    \caption{Regev's classical postprocessing procedure~\cite{Regev23}}\label{algo:regevpostprocessing}
\end{algorithm}

To analyze this algorithm in the way it is used in Algorithm \ref{algo:allpostprocessing}, we need to prove the following lemma:

\begin{lemma}\label{lemma:regevpostprocessingforus}
    (Compare with Lemma \ref{lemma:regevpostprocessing}, which also analyzes Algorithm \ref{algo:regevpostprocessing}) Let $\Lambda \subseteq \ZZ^d$, and let $T > 0$ be some norm bound. Assume we are given an arbitrary subset of $k = \gamma d$ out of $m = \alpha d$ independent samples of the form $$w_i = v_i + \delta_i,$$where each $v_i$ is a uniform sample from $\Lambda^\ast/\ZZ^d$ and $\delta_i$ is some additive error of magnitude at most $\delta$. Additionally, assume the following inequalities:
    \begin{itemize}
        \item $\big(\frac{e\alpha}{\alpha-\gamma}\big)^{\alpha-\gamma} \cdot 2^{-\gamma+1} < 1.$
        \item $(\gamma d + d)^{1/2} \cdot 2^{(\gamma+1)d/2} \cdot (\gamma d + 1)^{1/2} \cdot T < \delta^{-1} \cdot (4 \det \Lambda)^{-1/(\gamma d)} \cdot 2^{-\alpha/\gamma}/6.$
    \end{itemize}

    Then, with probability at least $1/4$, Algorithm \ref{algo:regevpostprocessing} succeeds when given the subset of $k$ samples as input i.e. it produces vectors $z_1, \ldots, z_l \in \Lambda$ that generate any element $u \in \Lambda$ with magnitude at most $T$.
\end{lemma}

We first set up some intermediate lemmas for the analysis necessary to prove Lemma \ref{lemma:regevpostprocessingforus}, closely mirroring the analysis by~\cite{Regev23}. The only nontrivial step in our adaptation is checking that this subset will generate $\Lambda^\ast/\ZZ^d$ with $\Omega(1)$ probability, which follows from the following lemma:
\begin{lemma}\label{lemma:lotsofgenerators}
    (Compare with Corollary 4.2 in~\cite{Regev23}) Let $G$ be an abelian group with a generating set (not necessarily minimal) of size $d$. Suppose we have $m = \alpha d$ independent and uniform samples from $G$, then with probability at least $1/2$, any $\gamma d$ of these samples will generate $G$.
\end{lemma}
\begin{proof}
    This follows from results by Acciaro~\cite{acciaro1996probability} on the probability that a certain number of samples will generate $G$. We defer details to Appendix \ref{sec:acciaroproof}.
\condqed \end{proof}

The remainder of our adaptation of Regev's analysis~\cite{Regev23} is straightforwrad, but we detail it below for completeness.

\begin{lemma}\label{lemma:regev4.3}
    (Compare with Lemma 4.3 in~\cite{Regev23}) Assume $v_1, \ldots, v_m$ are independent and uniform samples from $\Lambda^\ast/\ZZ^d$. Define a set $B' \subseteq [m]$ to be \emph{good} if $|B'| = \gamma d$.

    Then with probability at least $1/4$, for any nonzero $u \in \ZZ^d/\Lambda$ and any good $B'$, there exists an $i \in B'$ such that $\langle u, v_i \rangle \neq [-\epsilon, \epsilon] \bmod{1}$, where $\epsilon = (4\det \Lambda)^{-1/(\gamma d)} \cdot 2^{-\alpha/\gamma}/3$.
\end{lemma}
\begin{proof}
    We almost exactly follow the proof by Regev~\cite{Regev23}, except we also need to take a union bound over sets $B'$. We say $B'$ is \emph{generating} if the samples $\left\{v_i: i \in B'\right\}$ generate $\Lambda^\ast/\ZZ^d$.
    
    We need to upper bound the probability of there existing nonzero $u \in \ZZ^d/\Lambda$ and good $B'$ such that $\langle u, v_i \rangle \in [-\epsilon, \epsilon] \bmod{1}$ for all $i \in B'$. By a union bound, this is at most the sum of the probabilities of the following two events:
    \begin{itemize}
        \item $E_1$: There exists a good set $B'$ that is not generating.
        \item $E_2$: There exist nonzero $u \in \ZZ^d/\Lambda$ and $B'$ that is both good and generating such that $\langle u, v_i \rangle \in [-\epsilon, \epsilon]\bmod{1}$ for all $i \in B'$.
    \end{itemize}
    The abelian group $\Lambda^\ast/\ZZ^d$ has a generating set of size $d$ e.g. by taking a basis of $\Lambda^\ast$. Lemma \ref{lemma:lotsofgenerators} hence tells us that $\Pr[E_1] \leq 1/2$. So it remains to show that $\Pr[E_2] \leq 1/4$.

    For $E_2$, temporarily fix a nonzero $u \in \ZZ^d/\Lambda$ and a good and generating subset $B' \subseteq [m]$. As argued by~\cite{Regev23}, because $u$ is not the zero coset, when we sample uniform $v$ from $\Lambda^\ast/\ZZ^d$, $\langle u, v \rangle\bmod{1}$ must be a uniform sample from $\left\{0, 1/t, \ldots, (t-1)/t\right\}$ for some $t \geq 2$. There are two cases: if $t < 1/\epsilon$, then since $B'$ is generating there must exist $i \in B'$ such that $\langle u, v_i \rangle \bmod{1} \neq 0 \Rightarrow \langle u, v_i \rangle \notin [-\epsilon, \epsilon]\bmod{1}$, contradicting the hypothesis of $E_2$. If $t \geq 1/\epsilon$, then the probability of $\langle u, v \rangle \in [-\epsilon, \epsilon]\bmod{1}$ will be $(1 + 2\lfloor t\epsilon\rfloor)/t \leq 3\epsilon$. Since the samples are independent, the probability that $\langle u, v_i \rangle \in [-\epsilon, \epsilon]\bmod{1}$ for all $i \in B'$ is at most $(3\epsilon)^{|B'|} = (3\epsilon)^{\gamma d}$.

    To finish, we take a union bound. There are trivially at most $2^m$ choices of the set $B'$, and the number of choices of $u$ is at most $|\ZZ^d/\Lambda| = \det \Lambda$. Hence we have:
    \begin{align*}
        \Pr[E_2] &\leq \det \Lambda \cdot 2^m \cdot (3\epsilon)^{\gamma d} \\
        &= 1/4,
    \end{align*}
    which completes the proof of the lemma.
\condqed \end{proof}

\begin{lemma}\label{lemma:regev4.4}
    (Compare with Lemma 4.4 in~\cite{Regev23}) Assume Algorithm \ref{algo:regevpostprocessing} is given input as specified by Lemma \ref{lemma:ourpostprocessing}. Then both of the following are true:
    \begin{itemize}
        \item For any $u \in \Lambda$, there exists $u' \in \Lambda'$ whose first $d$ coordinates are equal to $u$ and whose norm is at most $||u||_2 \cdot (k+1)^{1/2} = ||u||_2 \cdot (\gamma d+1)^{1/2}$.
        \item With probability at least $1/4$ (over the randomness of $w_1, \ldots, w_m$), any nonzero $u' \in \Lambda'$ of norm $< \delta^{-1}\epsilon/2$ satisfies that its first $d$ coordinates are a nonzero vector in $\Lambda$, where $\epsilon = (4\det \Lambda)^{-1/(\gamma d)} \cdot 2^{-\alpha/\gamma}/3$.
    \end{itemize}
\end{lemma}
\begin{proof}
    Identical to the proof by~\cite{Regev23}; we can just plug in the result from Lemma \ref{lemma:regev4.3} instead of Lemma 4.3 in~\cite{Regev23}.
\condqed \end{proof}

\begin{lemma}\label{lemma:regev5.1}
    (Follows directly from Claim 5.1 in~\cite{Regev23}) In Algorithm \ref{algo:regevpostprocessing}, we have for all $i \in [l]$ that $||z_i||_2 \leq (d+k)^{1/2} \cdot 2^{(d+k)/2} \cdot (k+1)^{1/2} \cdot T$. Moreover, any vector in $\Lambda'$ of norm at most $(k+1)^{1/2} \cdot T$ is expressible as an integer linear combination of $z_1, \ldots, z_l$.
\end{lemma}
\begin{proof}
    See~\cite{Regev23}.
\condqed \end{proof}

We can now put everything together, following the argument at the end of~\cite{Regev23}:

\begin{proof}[Proof of Lemma \ref{lemma:regevpostprocessingforus}]
    Consider any $u \in \Lambda$ such that $||u||_2 \leq T$. Then by Lemma \ref{lemma:regev4.4}, there exists $u' \in \Lambda'$ whose first $d$ coordinates are equal to $u$ and with norm at most $(k+1)^{1/2} \cdot ||u||_2 \leq (k+1)^{1/2} \cdot T$. It follows by Lemma \ref{lemma:regev5.1} that there exist integers $\alpha_1, \ldots, \alpha_l$ such that:
    \begin{align*}
        u' &= \alpha_1 z_1 + \ldots + \alpha_lz_l \\
        \Rightarrow u &= \alpha_1 z'_1 + \ldots + \alpha_lz'_l,
    \end{align*}
    by equating the first $d$ coordinates. It hence remains to show that $z'_i \in \Lambda$ for all $i \in [l]$. To do this, note that Lemma \ref{lemma:regev5.1} also tells us for any $i \in [l]$ that:
    \begin{align*}
        ||z_i||_2 &\leq (d+k)^{1/2} \cdot 2^{(d+k)/2} \cdot (k+1)^{1/2} \cdot T \\
        &= (\gamma d + d)^{1/2} \cdot 2^{(\gamma+1)d/2} \cdot (\gamma d + 1)^{1/2} \cdot T \\
        &< \delta^{-1} \cdot (4 \det \Lambda)^{-1/(\gamma d)} \cdot 2^{-\alpha/\gamma}/6 \\
        &= \delta^{-1}\epsilon/2.
    \end{align*}
    By Lemma \ref{lemma:regev4.4}, it follows that with probability at least $1/4$, each $z'_i$ must indeed be in $\Lambda$.
\condqed \end{proof}

\section{Proof of Lemma \ref{lemma:lotsofgenerators}}\label{sec:acciaroproof}

The probability that $t$ independent and uniform samples from an abelian group $G$ generate the group can be calculated exactly from results by Acciaro~\cite{acciaro1996probability}. We do this in the following two lemmas. For an abelian group $G$ and integer $r$, let $\alpha_r(G)$ denote the probability that $r$ uniform and independent samples from $G$ generate $G$.

\begin{lemma}\label{lemma:acciaropgroup}
    (Lemma 4 in~\cite{acciaro1996probability}, also stated by~\cite{pomerance2002expected}) Let $G$ be a finite abelian group with order equal to some power of a prime $p$, and let $r$ be the minimal number of generators for $G$. Then for any $t \geq r$, we have $$\alpha_t(G) = \prod_{i = t-r+1}^t (1-p^{-i}).$$
\end{lemma}
\begin{proof}
    This is not exactly the form that Acciaro's result is stated in, but Pomerance~\cite{pomerance2002expected} restates the result in this form in equation 2 of his paper.
\condqed \end{proof}

\begin{lemma}\label{lemma:acciarogeneralgroup}
    Let $G$ be a finite abelian group, and suppose there exists a generating set of size $r$ for $G$ (there may also be a strictly smaller generating set). Then for any $t \geq r$, we have $$\alpha_t(G) \geq \prod_{p \mid |G|} (1 - p^{r-1-t})^r.$$
\end{lemma}
\begin{proof}
    First, note that by writing $G$ as a direct product of cyclic abelian groups of prime power order and then combining groups whose orders are powers of the same prime, we can write $$G = G_{p_1} \times G_{p_2} \times \ldots \times G_{p_k},$$where $p_1, \ldots, p_k$ are the distinct primes dividing $|G|$. Each $G_{p_i}$ may not be cyclic, but its order will be a power of $p_i$. By Corollary 4 in~\cite{acciaro1996probability} and noting that the orders of the $G_{p_i}$'s are pairwise coprime, we have $$\alpha_t(G) = \prod_{p \mid |G|} \alpha_t(G_p).$$
    It therefore suffices to show for each $p$ that $\alpha_t(G_p) \geq (1 - p^{r-1-t})^r$. To do this, fix a $p$ and let $s$ be the minimal number of generators for $G_p$. We have $s \leq r$ (to see this, note that any $r$ generators for $G$ project down into a generating set of $r$ elements for $G_p$, so a minimal generating set for $G_p$ has size at most $r$). Now we can use Lemma \ref{lemma:acciaropgroup} to proceed as follows:
    \begin{align*}
        \alpha_t(G_p) &= \prod_{i = t-s+1}^t (1 - p^{-i}) \\
        &\geq \prod_{i = t-s+1}^t (1 - p^{r-1-t}) \text{ (since $i \geq t-s+1 \geq t-r+1$)} \\
        &= (1 - p^{r-1-t})^s \\
        &\geq (1 - p^{r-1-t})^r.
    \end{align*}
\condqed \end{proof}

We can now complete the proof of Lemma \ref{lemma:lotsofgenerators}. Let $k = \gamma d$. We will bound $\alpha_k(G)$ by applying Lemma \ref{lemma:acciarogeneralgroup} with $t = k$ and $r = d$. We hence have:
\begin{align*}
    \alpha_k(G) &\geq \prod_{p \mid |G|} (1 - p^{d-1-k})^d \\
    &\geq \prod_{\text{all primes }p} (1 - p^{d-1-k})^d \\
    &= \frac{1}{\zeta(k+1-d)^d},
\end{align*}
where we have used the Euler product formula for the Riemann zeta function $\zeta$. Since $k+1-d$ is a real number $> 1$, we have:
\begin{align*}
    \zeta(k+1-d) &= \sum_{x = 1}^\infty \frac{1}{x^{k+1-d}} \\
    &\leq 1 + \frac{1}{2^{k+1-d}} + \int_2^\infty \frac{1}{x^{k+1-d}}dx \\
    &= 1 + \frac{1}{2^{k+1-d}} + \big[\frac{x^{-k+d}}{k-d}\big]_\infty^2 \\
    &= 1 + 2^{-k+d-1} + \frac{2^{-k+d}}{k-d} \\
    &\leq 1 + 2^{-k+d} \\
    \Rightarrow \frac{1}{\zeta(k+1-d)} &\geq 1 - 2^{-k+d} \\
    \Rightarrow \frac{1}{\zeta(k+1-d)^d} &\geq (1 - 2^{-k+d})^d \\
    &\geq 1 - d \cdot 2^{-k+d} \text{ (by Bernoulli's inequality).}
\end{align*}
Therefore, the probability that $k$ uniform and independent samples from $G$ do not generate $G$ is $$1 - \alpha_k(G) \leq d \cdot 2^{-k+d}.$$
Now to finish, let us take a union bound. The probability that out of $m = \alpha d$ uniform and independent samples, there exists a subset of size $k = \gamma d$ that does not generate $G$ is at most:
\begin{align*}
    \binom{m}{k} \cdot d \cdot 2^{-k+d} &= \binom{m}{m-k} \cdot d \cdot 2^{-k+d} \\
    &\leq \left(\frac{em}{m-k}\right)^{m-k} \cdot d \cdot 2^{-k+d} \\
    &= \left(\frac{e\alpha}{\alpha-\gamma}\right)^{(\alpha-\gamma)d} \cdot d \cdot 2^{(-\gamma+1)d} \\
\end{align*}
which is $< 1/2$ for $d$ sufficiently large, provided that $$\left(\frac{e\alpha}{\alpha-\gamma}\right)^{\alpha-\gamma} \cdot 2^{-\gamma+1} < 1.$$

This completes the proof of Lemma \ref{lemma:lotsofgenerators}.\qed

\section{Efficient Modular Fibonacci Exponentiation}\label{sec:modexp}

In Section \ref{sec:oracle}, we adapted ideas by Kaliski~\cite{kal17} to show how Fibonacci exponentiation could be used to efficiently compute a product of exponents on a quantum computer. It then follows that the same technique would also be helpful to compute a single exponent. We capture this in the following standalone result, and then discuss its applicability to other quantum algorithms.

As in Theorem \ref{thm:mainresult}, we assume throughout that we have a multiplication circuit that maps $$\ket{a}\ket{b}\ket{t}\ket{0^S} \mapsto \ket{a}\ket{b}\ket{(t + ab) \bmod N} \ket{0^S}$$ where $a$, $b$, $t$ are all $n$-bit integers and $0 \leq a, b, t < N$.

\begin{lemma}\label{lemma:modexp}
    There exists a quantum circuit using $O(mG + m^2)$ gates such that, for all $n$-bit integers $a$ coprime to and reduced modulo $N$ and for all $m$-bit integers $z$, it will map $$\ket{a}\ket{a^{-1}\bmod N}\ket{z}\ket{0^{M}} \mapsto \ket{a}\ket{a^{-1}\bmod N}\ket{z}\ket{a^z\bmod N}\ket{0^{M-n}},$$
    where $M = S + O(m+n)$ is the number of ancilla qubits.
\end{lemma}

\begin{proof}
    Our algorithm proceeds very similarly to Algorithm \ref{algo:fulloracle}, so we just outline the key steps below and point out the space and size costs:
    \begin{enumerate}
        \item Construct $z_j \in \left\{0, 1\right\}$ for $j \leq \frac{m}{\log\phi}$ such that $z = \sum_{j} z_j F_j$. This can be done using the greedy algorithm in Lemma \ref{lemma:zeckendorf}. This uses $O(m)$ ancilla qubits and comprises $O(m)$ additions/subtractions/comparisons on $m$-bit integers, for a total of $O(m^2)$ gates. The expression we want to compute is now $$\prod_{j} a^{z_jF_j} \bmod{N} = \prod_{j} c_j^{F_j} \bmod{N},$$ where $c_j = a^{z_j} \in \left\{1, a\right\}$.
        \item Follow Algorithm \ref{algo:basicfibonacci} to now compute the desired product. This uses $S+5n = S+O(n)$ ancilla qubits and $O(mG)$ gates. Note importantly that since all the $c_j$'s are either 1 or $a$, we only need $\ket{\psi(a)} = \ket{a}\ket{a^{-1}\bmod N}$, which we are given as part of the input.
        \item Copy the final output $\ket{a^z \bmod{N}}$ to an $n$-bit register, and uncompute all previous operations to restore the original inputs. This requires another $O(n)$ ancilla qubits and essentially doubles the gate complexity (there is another $O(n)$ gates for the copying, but this will be dominated by $G$).
    \end{enumerate}
    Hence the overall ancilla cost is $S + O(m+n)$ and the circuit size is $O(mG + m^2)$, as desired.
\condqed \end{proof}
We make some observations about our result below:
\begin{enumerate}
    \item Exponentiation via the square-and-multiply algorithm involves $O(m)$ multiplications and hence requires $O(mG)$ gates. However, all intermediate results will have to be written out in separate registers, hence increasing the space demand to $S + O(mn)$ which is strictly worse than that of our construction.

    We note that our result is not strictly better than square-and-multiply in the case that $m \gg G$; in this case, the $O(m^2)$ overhead from computing the Zeckendorf representation of the exponent will dominate the cost of the multiplications. This is not of concern when $m = O(n)$ or $m = O(\sqrt{n})$, as is the case in Shor's and Regev's algorithms respectively.
    \item Precomputation-based approaches will be more efficient in situations where the set of possible values of $a$ is small (or generated by a small subset of $\ZZ_N^\ast$). This is what happens in Shor's algorithm.
    \item The need for $a^{-1} \bmod N$ to be available is potentially inconvenient. Kaliski's algorithm~\cite{kal17} does not require this, but it appears necessary when adapting to the quantum setting. Ideally, we would have a case like our implementation of Regev's algorithm where it is relatively straightforward to compute the inverses of all the bases we are interested in. However, in the worst case, one could use the extended Euclidean algorithm to compute $\ket{a^{-1}\bmod{N}}$, which may or may not be an acceptable overhead depending on the application.
\end{enumerate}

\section{Applying Our Error Detection Procedure to Calculate Discrete Logarithms}\label{sec:eccdl}

In this section, we briefly outline the adaptation by Eker{\aa} and G{\"a}rtner~\cite{ekeragartner} of Regev's factoring algorithm~\cite{Regev23} to the discrete logarithm problem modulo prime $p$, and highlight that our error tolerance result (specifically, Lemma \ref{lemma:ourpostprocessing}) also directly applies to their algorithm.

The algorithm follows a very similar blueprint to~\cite{Regev23}. It runs a quantum circuit of size $O(n^{1/2} \cdot G + n^{3/2})$ (where $G$ is the size of an integer multiplication circuit as in Theorem \ref{thm:mainresult}) for $O(n^{1/2})$ iterations to obtain several samples from the dual of some lattice. Because of the large number of runs of the quantum circuit, this algorithm also stands to benefit in principle from a classical postprocessing procedure that tolerates a constant fraction of unsuccessful runs.

Let the inputs for the discrete logarithm problem be a generator $g$ and target value $x$. Let $b_1, \ldots, b_{d-2}$ be some small integers as in the factoring algorithm by~\cite{Regev23}. Then Eker{\aa} and G{\"a}rtner~\cite{ekeragartner} construct the following lattice, for a generator $g$ and target value $x$:
$$\mathcal{L}_{x, g} = \left\{(z_1, \ldots, z_d) \in \ZZ^d \mid x^{z_{d-1}}g^{z_d}\prod_{i = 1}^{d-2} b_i^{z_i} \equiv 1\text{ }(\bmod\text{ }p)\right\}.$$
For this lattice, they make a similar number-theoretic assumption to that proposed by~\cite{Regev23}:
\begin{conjecture}\label{ekeraconjecture} 
There exists a basis for $\mathcal{L}_{x, g}$, where all basis elements have $\ell_2$ norm at most $T = 2^{C\sqrt{n}}$, for some given constant $C > 0$.
\end{conjecture}
They then argue via Lemma \ref{lemma:regevpostprocessing} that Regev's classical postprocessing procedure~\cite{Regev23} will find a collection of vectors in $\mathcal{L}_{x, g}$ that generates all elements in $\mathcal{L}_{x, g}$ of magnitude at most $T$. Under Conjecture \ref{ekeraconjecture}, this must mean they generate $\mathcal{L}_{x, g}$ itself, hence we have a basis for $\mathcal{L}_{x, g}$~\cite{ekeragartner}. Now one can directly solve to find elements $(z_1, \ldots, z_d) \in \mathcal{L}_{x, g}$ such that $z_1 = z_2 = \ldots = z_{d-2} = 0$ and thus recover the discrete logarithm~\cite{ekeragartner}.

By plugging in the postprocessing procedure in our Lemma \ref{lemma:ourpostprocessing} in place of Lemma \ref{lemma:regevpostprocessing}, one can readily obtain a result analogous to Theorem \ref{thm:mainthm} for the discrete logarithm problem modulo $p$.

\section{Calculating Parameters}

\subsection{Justifying the Choice of Parameters in~\cite{Regev23}}\label{sec:paramcheck}

Plugging everything into the inequality stated in Lemma \ref{lemma:regevpostprocessing} implies that we require:
\begin{align*}
    & (2d+4)^{1/2}  \cdot 2^{d+2} \cdot (d+5)^{1/2} \cdot T < \delta^{-1} (4 \cdot 2^n)^{-1/(d+4)}/6 \\
    & \Leftrightarrow (2d+4)^{1/2} \cdot 2^{d+2} \cdot (d+5)^{1/2} \cdot 2^{C\sqrt{n}} < 2^{(A-o(1))n/d} \cdot (4 \cdot 2^n)^{-1/(d+4)}/6 \\
    & \Leftrightarrow (A-o(1))n/d > \log_2(6d^{1/2} \cdot (d+2)^{1/2} \cdot 2^{d+2} \cdot (d+5)^{1/2} \cdot 2^{C\sqrt{n}} \cdot 2^{(n+2)/(d+4)}) \\
    &\hspace{1in} = o(\sqrt{n}) + C\sqrt{n} + (d+2) + \frac{n+2}{d+4} \\
    &\hspace{1in} = o(\sqrt{n}) + C\sqrt{n} + \sqrt{n} + \sqrt{n} \\
    &\hspace{1in} = (C + 2 + o(1))\sqrt{n},
\end{align*}
thus taking $A \geq C+2+o(1)$ is sufficient.

\subsection{Justifying the Choice of Parameters in Theorem \ref{thm:mainthm}}\label{sec:ourparamcheck}

Since $d = \sqrt{n}$, we can plug in $T = 2^{Cd}$ and also use the bound $\det \mathcal{L} < 2^n$ and set $\delta = 2^{(-A+o(1))d}$. With these parameters, the special inequality becomes:
\begin{align*}
    (\gamma d + d)^{1/2} \cdot 2^{(\gamma+1)d/2} \cdot (\gamma d + 1)^{1/2} \cdot T &< \delta^{-1} \cdot (4 \det \mathcal{L})^{-1/(\gamma d)} \cdot 2^{-\alpha/\gamma}/6 \\
    \Leftarrow (\gamma d + d)^{1/2} \cdot 2^{(\gamma+1)d/2} \cdot (\gamma d + 1)^{1/2} \cdot 2^{Cd} &< 2^{(A-o(1))d} \cdot (4 \cdot 2^n)^{-1/(\gamma d)} \cdot 2^{-\alpha/\gamma}/6 \\
    &= 2^{(A - 1/\gamma - o(1))d} \\
    \Leftrightarrow 2^{(C + \gamma/2 + 1/2 + o(1))d} &< 2^{(A - 1/\gamma - o(1))d},
\end{align*}
which holds for $A \geq C + \frac{\gamma^2+\gamma+2}{2\gamma} + o(1)$.

\section{Proof of Lemma \ref{lemma:uniformiswellspread}}\label{sec:uniformiswellspread}

Let us temporarily fix coefficients $a_i$ such that at least one is nonzero. Let us also fix a vector $v \in \Lambda^\ast/\ZZ^d$ such that $d(\sum_{i = 1}^k a_i\epsilon_i, v) \leq 2^{(-2A+\alpha-\gamma+3+o(1))n/(2d)}.$ We will take a union bound over all such choices at the end.

Since at least one $a_i$ is nonzero and the $a_i$'s are all integers, it is easy to see that $\sum_{i = 1}^k a_i\epsilon_i$ will be uniformly distributed over $\RR^d/\ZZ^d$. Hence $\sum_{i = 1}^k a_i\epsilon_i - v$ will also be a uniformly distributed sample from $\RR^d/\ZZ^d$.

We want to bound the probability that the norm of this sample on the torus is at most $2^{(-2A+\alpha-\gamma+3+o(1))n/(2d)}$. This is at most the probability that each coordinate of our sample is that close to an integer, which is at most:
\begin{align*}
    \left(2 \cdot 2^{(-2A+\alpha-\gamma+3+o(1))n/(2d)}\right)^d &= 2^{(-2A+\alpha-\gamma+3+o(1))n/2}.
\end{align*}
Finally, we take our union bound over the following:
\begin{itemize}
    \item The choice of $v \in \Lambda^\ast/\ZZ^d$. There are at most $|\Lambda^\ast/\ZZ^d| = \det \Lambda \leq N < 2^n$ such choices.
    \item The choice of the $a_i$'s. There are at most $2^{(\alpha-\gamma+3+o(1))n/(2d) \cdot k} \leq 2^{((\alpha-\gamma-1)(\alpha-\gamma+3)+o(1))n/2}$ such choices.
\end{itemize}
Hence the probability of problematic integers $a_1, \ldots, a_k$ existing is at most:
\begin{align*}
    2^n \cdot 2^{((\alpha-\gamma-1)(\alpha-\gamma+3)+o(1))n/2} \cdot 2^{(-2A+\alpha-\gamma+3+o(1))n/2} &= 2^{(-2A + (\alpha-\gamma)(\alpha-\gamma+3) + 2)n/2},
\end{align*}
which is negligible provided $A > \frac{(\alpha-\gamma)(\alpha-\gamma+3) + 2}{2}$.\qed

\section{Proof of Correctness in Lemma \ref{lemma:fibonacci}}\label{sec:fibonacciproof}

It remains to show the final claim made in the proof of Lemma \ref{lemma:fibonacci} by induction. First we check the base case i.e. the first two rounds where $j = K, K-1$. In round $K$, $(x_1, x_2)$ evolves as follows: $(1, 1) \rightarrow (1, 1) \rightarrow (c_K, 1) \rightarrow (1, c_K)$. Then in round $K-1$, it evolves as follows: $(1, c_K) \rightarrow (c_K, c_K) \rightarrow (c_{K-1}c_K, c_K) \rightarrow (c_K, c_{K-1}c_K)$. This is consistent with our claim for $j = K-1$.

For the inductive step, assume the current round is $j$, and the previous round (indexed by $j+1$) has ended according to our claim. The current state is hence:
$$(x_1, x_2) = \left(\prod_{i = j+2}^K c_i^{F_{i-j-1}}, \prod_{i = j+1}^K c_i^{F_{i-j}}\right).$$
After the first update, $x_1$ becomes:
\begin{align*}
    \prod_{i = j+2}^K c_i^{F_{i-j-1}} \cdot \prod_{i = j+1}^K c_i^{F_{i-j}} &= c_{j+1}^{F_1} \cdot \prod_{i = j+2}^K c_i^{F_{i-j-1} + F_{i-j}} \\
    &= c_{j+1} \cdot \prod_{i = j+2}^K c_i^{F_{i-j+1}} \\
    &= \prod_{i = j+1}^K c_i^{F_{i-j+1}}.
\end{align*}
After the second update, it becomes:
\begin{align*}
    c_j \cdot \prod_{i = j+1}^K c_i^{F_{i-j+1}} &= \prod_{i = j}^K c_i^{F_{i-j+1}}.
\end{align*}
Swapping the registers now gives us the state $$(x_1, x_2) = \left(\prod_{i = j+1}^K c_i^{F_{i-j}}, \prod_{i = j}^K c_i^{F_{i-j+1}}\right)~,$$ completing our induction.

After the $j = 1$ stage, our state is hence $(\prod_{i = 2}^K c_i^{F_{i-1}}, \prod_{i = 1}^K c_i^{F_{i}})$, as desired.\qed

\section{Decomposing Positive Integers as a Sum of Fibonacci Numbers}\label{sec:zeckendorf}

It was shown by \cite{zeckendorf} that any positive integer has a unique decomposition as a sum of Fibonacci numbers, if we enforce that no two of the Fibonacci numbers should be consecutive. We only need a weaker property, which we restate and prove here:

\begin{lemma}\label{lemma:zeckendorfgreedy}
    Consider the following algorithm, that takes as input a non-negative integer $t_0 < F_{k+1}$. Initialize $t \leftarrow t_0$ and $S$ to be the empty set. Now for $j = k, k-1, \ldots, 1$, check whether $t \geq F_j$. If it is, update $t \leftarrow t - F_j$ and add $F_j$ to $S$.

    Then at the end of the algorithm, we will have $t = 0$ and the elements of $S$ adding to $t_0$.
\end{lemma}
\begin{proof}
    First, observe that the algorithm clearly ensures that $t \geq 0$ at all times, and that $t$ and the elements of $S$ together add to $t_0$. Hence it suffices to show that at the end of the algorithm, we will have $t = 0$.
    
    We do this by strong induction on $k$. The base case $k = 1$ follows trivially since $t < F_2 = 1$ forces $t = 0$, so we are already done. We also consider the base case $k = 2$; in this case we have $t < F_3 = 2$. If $t = 0$ then we are already done. Else we have $t = 1 = F_2$. The algorithm will hence add $F_2$ to $S$ and $t$ will become 0.

    Now for the inductive step, consider $k > 2$. We have two cases:
    \begin{itemize}
        \item If $t < F_k$, then in the $j = k$ round nothing will happen, and we simply reduce to the $k-1$ case.
        \item If $t \geq F_k$, then in the $j = k$ round $t$ will be replaced by $t - F_k < F_{k+1} - F_k = F_{k-1}$. Hence nothing will happen in the $j = k-1$ round, and at this point we have reduced to the $k-2$ case.
    \end{itemize}
    Either way, the conclusion follows by induction.
\condqed \end{proof}

\end{document}